\documentclass[a4paper,12pt]{article}
\usepackage{mathtools}
\usepackage{amsfonts}
\usepackage{amsthm}
\usepackage{amsbsy}
\usepackage{enumerate}
\usepackage{bbm}
\usepackage{fullpage}
\usepackage{color}
\usepackage{microtype}
\usepackage{hyperref}
\usepackage{graphbox}
\usepackage{float}

\usepackage[mathlines]{lineno}
\theoremstyle{plain}
\newtheorem{thm}{Theorem}
\newtheorem{lemma}[thm]{Lemma}
\newtheorem{prop}[thm]{Proposition}
\newtheorem{cor}[thm]{Corollary}

\theoremstyle{definition}

\newcommand{\bbR}{\mathbb{R}}

\newcommand{\cA}{\mathcal{A}}
\newcommand{\cC}{\mathcal{C}}
\newcommand{\cP}{\mathcal{P}}
\newcommand{\cS}{\mathcal{S}}

\newcommand{\rA}{\mathrm{A}}
\newcommand{\rB}{\mathrm{B}}
\newcommand{\rC}{\mathrm{C}}
\newcommand{\rD}{\mathrm{D}}

\newcommand{\rI}{\mathrm{I}}

\newcommand{\rS}{\mathrm{S}}

\newcommand{\epsf}{\varepsilon_{\mathrm{fast}}}

\newcommand{\epsx}{\varepsilon_{\mathrm{x}}}

\DeclareMathOperator{\atan}{\mathrm{arctan}}

\begin{document}

\title{Contrasting chaotic and stochastic forcing: tipping windows and attractor crises}

\author{Peter Ashwin, Julian Newman and Raphael R\"omer\\
Department of Mathematics and Statistics,\\ University of Exeter, Exeter EX4 4QF, UK}

\maketitle

\begin{abstract}
Nonlinear dynamical systems subjected to a combination of noise and time-varying forcing can exhibit sudden changes, critical transitions or tipping points where large or rapid dynamic effects arise from changes in a parameter that are small or slow. Noise-induced tipping can occur where extremes of the forcing causes the system to leave one attractor and transition to another. If this noise corresponds to unresolved chaotic forcing, there is a limit such that this can be approximated by a stochastic differential equation (SDE) and the statistics of large deviations determine the transitions. Away from this limit it makes sense to consider tipping in the presence of chaotic rather than stochastic forcing. In general we argue that close to a parameter value where there is a bifurcation of the unforced system, there will be a {\em chaotic tipping window} outside of which tipping cannot happen, in the limit of asymptotically slow change of that parameter. This window is trivial for a stochastically forced system. Entry into the chaotic tipping window can be seen as a boundary crisis/non-autonomous saddle-node bifurcation and corresponds to an exceptional case of the forcing, typically by an unstable periodic orbit. We discuss an illustrative example of a chaotically forced bistable map that highlight the richness of the geometry and bifurcation structure of the dynamics in this case. If a parameter is changing slowly we note there is a {\em dynamic tipping window} that can also be determined in terms of unstable periodic orbits.
\end{abstract}

\tableofcontents

\section{Introduction} \label{Intro}

For complex nonlinear systems subjected to a time-varying input (that we call forcing in general), there may be critical transitions (or tipping points) where a small change in input leads to a large change in system state.  In many contexts both intrinsic noisy variability and slow drift variability in parameters have a strong deterministic component and either of these may give rise to abrupt changes via noise-induced or bifurcation-induced tipping \cite{kuehn2011mathematical,kuehn2013mathematical}. Moreover, if the slow drift is not slow enough, rate-induced tipping may also appear \cite{Ashwinetal:2012,Ashwinetal:2017,Wieczorek:2023}. 

Indeed, such effects may already appear within a subsystem of an autonomous (unforced) nonlinear multiscale system because of other timescales present within the system. As an example, noise in climate models is often seen to be result of fast forcing by turbulent fluid dynamics of ocean and atmosphere, while variation appears due to slow forcing such as changes in sea level or land ice. Each of these fast noise and slow drift processes has a distinct effect on the dynamics of the climate. This relates to Hasselmann's programme to understand weather as ``noise'' added to climate: see for example \cite{Lucarini2023theoretical} for a recent review, and is particularly relevant in case where critical transitions/tipping points and chaotic variability are both present - see for example \cite{lucarini2017edge,bodai2012annual,ashwin2021physical}. In many problems of practical importance there are no clear timescale separations, meaning that we are not at the limit where stochastic forcing is appropriate.

The paper is organized as follows. In Sec.~\ref{sec:limits} we briefly review relations between stochastic and chaotic forced models with tipping, and describe the chaotic tipping window that appears in the chaotically forced case. Sec.~\ref{sec:Toy} discusses a discrete-time chaotically forced bistable system where we analyse the structure of the dynamics and geometry within phase and parameter space. Implications of slow variation of parameters in this model are discussed in Sec.~\ref{sec:ramped}, where we introduce the dynamic tipping window. In Sec.~\ref{sec:discuss} discuss possible implications and lessons for more general cases.

\section{Fast chaos limits and stochastic forcing}
\label{sec:limits}

Let us consider a system modelled by a multiscale ODE of the form
\begin{equation}
\begin{aligned}
 \frac{dx}{dt} = & f(x,y,\beta(t)) + \frac{1}{\varepsilon_x}f_0(x,y,\beta(t))  \\
  \frac{dy}{dt} = & \frac{1}{\epsf} g(x,y,\beta(t)) 
 \end{aligned}
 \label{eq:2scalenon-auton}
\end{equation}
where $x\in\bbR^d$ is the system we wish to characterise, $y\in\bbR^f$ is a chaotic subsystem representing the variability, $\beta(t)\in\bbR^p$ are parameters that may also be varying.

The quantity $\epsf>0$ represents a ratio between the timescales fast and slow processed. The terms multiplied by $1/\varepsilon_{x}$ represent the amplitude of fast fluctuations on $x$, with an assumption that $1\gg\epsx\gg\epsf$. Assume that $f_0$ averages to zero over the fast $y$ dynamics. For fixed $\beta$, reduction of systems derived from \eqref{eq:2scalenon-auton} to \eqref{eq:forcedsde} is possible for some cases of large forcing by fast chaos. For example consider $\epsx=\varepsilon^{1/2}$ and $\epsf=\varepsilon$ then
\begin{equation}
\begin{aligned}
\frac{dx}{dt} = &  f(x,\beta)+ \frac{1}{\sqrt{\varepsilon}}f_0(x,y,\beta)  \\
\frac{dy}{dt} = & \frac{1}{\varepsilon} g(y,\beta)
\end{aligned}
\label{eq:forcedchaos}
\end{equation}
where $y$ is chaotic and the average of $f_0$ over this attractor $\langle f_0\rangle_y=O(\varepsilon)$. Recent limit theorem results~\cite{Gaspard:2015,Kelly:2017} (see the review \cite{gottwald2017}) can be applied to show that trajectories of $x$ of \eqref{eq:forcedchaos} are well-approximated in the limit $\varepsilon_f\rightarrow 0$ as solutions of the It\^{o} SDE
\begin{equation}
dx = f(x,\beta(t))dt + \sigma(x) dW_t
\label{eq:forcedsde} 
\end{equation}
where $(W_t)_{t \in \mathbb{R}}$ is a Wiener process. If the SDE (\ref{eq:forcedsde}) is written informally as
\begin{equation}
\frac{dx}{dt} = f(x,\beta(t)) + \sigma(x)\eta(t)
\label{eq:forcedsdeinformal}
\end{equation}
then it can be viewed as a deterministic nonlinear system 
\begin{equation}
\frac{dx}{dt} =f(x,\beta)
\label{eq:ode}
\end{equation} 
subjected to white noise forcing $\eta(t)$ (a ``time derivative'' of $W_t$) and slow variation of parameters $\beta(t)$.

\subsection{Tipping points and the chaotic tipping window}

Systems of the form (\ref{eq:2scalenon-auton}) and (\ref{eq:forcedsde}) allow one to discuss the variables $x$ of interest in terms how they are forced by the other variables. Tipping events (large or rapid changes in state) can arise in \eqref{eq:forcedsde} for various reasons \cite{Ashwinetal:2012}. These reasons include \emph{bifurcation-induced} and \emph{noise-induced} tipping that can be attributed to loss of stability on slowly changing $\beta$ or because of extremes of $\eta$. More rapid changes of parameters $\beta(t)$ can give rise to \emph{rate-induced} tipping and in general abrupt transitions may be attributable to a mixture of these effects \cite{Ashwinetal:2012}.

Because the white noise in \eqref{eq:forcedsde} is unbounded, with small but finite probability the noise can drive any system to exceed an arbitrarily large magnitude in an arbitrarily short-duration time-interval. This means that in this context noise-induced tipping between any attractor is always possible on any timescale. The theory of large deviations for such systems \cite{berglund2006noise,Freidlin2012} gives a powerful set of tools to approximate the statistics of such rare events, especially in the case of asymptotically small $\sigma\ll 1$ but the presence of white noise \emph{destroys} any multistability of \eqref{eq:ode}: there is a non-zero chance of transition between neighbourhoods of any pairs of attractors.

However, except in the limit of asymptotically fast chaotic forcing, the noise component of \eqref{eq:2scalenforced} is never truly a Wiener process. Indeed, it is often bounded: the extremes that drive tipping events cannot exceed certain levels.
Consider a model \eqref{eq:2scalenon-auton}  with bounded noise and forcing by a deterministic chaotic system. We find that tipping event induced by the bounded noise can still occur, but there the boundedness means that there are limited regions where tipping is possible. The bifurcations to tipping associated with bounded noise can be understood from the viewpoint of set-valued dynamics \cite{lamb2015topological,kuehn2018early}. We consider a chaotically forced system to the dynamics:
\begin{equation}
\begin{aligned}
\frac{dx}{dt} = &\ f(x,\beta(t)) + \sigma f_0(y)  \\
\frac{dy}{dt} = &\ g(y).
\end{aligned}
\label{eq:2scalenforced}
\end{equation}
The dynamical behaviour of the $x$ dynamics of \eqref{eq:2scalenforced} is forced by chaotic behaviour $y$ and a slow drift $\beta$
both assumed independent from $x$ and from each other. For fixed $\beta$ we have the {\em frozen system} with chaotic forcing
\begin{equation}
\begin{aligned}
\frac{dx}{dt} = &\ f(x,\beta) + \sigma f_0(y)  \\
\frac{dy}{dt} = &\ g(y)
\end{aligned}
\label{xaut2}
\end{equation}
where we assume the initial $y(0)$ is chosen to be typical with respect to some ergodic invariant measure $m$ for the $y$ dynamics: we can view this as a non-autonomous system
\begin{equation}
\begin{aligned}
\frac{dx}{dt} = &\ f(x,\beta) + \sigma f_0(y(t)).
\end{aligned}
\label{xna}
\end{equation}

Suppose this system with $\sigma=0$ has a branch of stable equilibria for $\beta<\beta^*$ and a bifurcation at some $\beta^*$ such that there is no nearby attractor for $\beta>\beta^*$. The addition of chaotic forcing with $\sigma>0$ will create a {\em chaotic tipping window} $(\beta_-,\beta_+)$  depending on $\sigma$ such that the chaotic forcing outside the window cannot influence the tipping, such that this window limits to $\beta^*$ as $\sigma\rightarrow 0$.

The entry to the chaotic tipping window at $\beta^-$ is a {\em boundary crisis} of an attractor of \eqref{xaut2}. Note that in the originally a crisis is an attractor hitting the stable manifold of an unstable periodic orbit on the basin boundary \cite{grebogi1983crises} but a more general notion of boundary crisis when attractor hitting its basin boundary can occur at a chaotic saddle \cite{hong2004achaotic,lai2011transient}. Beyond the exit of the chaotic tipping window at $\beta^+$ the dynamics of the chaotic system $y$ can no longer delay tipping and there will be rapid escape from a neighbourhood of where the stable branch previously existed. Within the chaotic tipping window the behaviour will depend on the precise details of the chaotic trajectory $y(t)$. In particular, the chaotic attractor will typically contain many unstable periodic orbits (UPOs) and other exceptional measures whose statistics do not reflect that of the invariant measure $m$.

Let us consider the set $\cS(m)$ of invariant measures for the $y$ dynamics that are supported within the support of $m$. For our system we can characterise the chaotic tipping window as the smallest window that contains 
$$
\beta_c(s):=\sup\{\beta~:~ \mbox{the system has more than one attractor} \mbox{ when forced by }s \}.
$$
for all $s\in\cS(m)$. Determining the end points of this window is hence a problem of ergodic optimization: we need to find
$$
\beta_-=\inf \{ \beta_c(s)~:~s\in\cS(m)\},~
~~\beta_+=\sup \{ \beta_c(s)~:~s\in\cS(m)\}.
$$
For typical systems we expect the measures realising these extremes will be periodic \cite{jenkinson2019ergodic}. However, because $\beta_c(s)$ is not simply an average but the boundary of a region of multistability, it does not seem to be obvious how to directly apply these results. We note that the threshold on entry to the chaotic tipping window can also be understood in terms of a \emph{nonautonomous saddle-node bifurcation}~\cite{Anagnostopoulou:2012} of the system \eqref{xna}.

\section{Dynamics and geometry for a chaotically forced bistable map}
\label{sec:Toy}

A noise-forced tipping of an ODE forced by a chaotic component \eqref{eq:forcedchaos} (as for example in \cite{NewmanAshwin:2023}) will be very hard to understand in detail and so for illustration we examine an iterated bistable map forced by a prototypical chaotic map. The example we consider is the skew product dynamical system on $(x,\theta)\in \bbR\times [0,2\pi)$ defined by
\begin{equation}
\label{eq:map}
\begin{aligned}
x_{n+1} = & f(x_n,\theta_n,\beta,\sigma)\\
\theta_{n+1}=  &g(\theta_n)
\end{aligned}
\end{equation}
where $g(\theta)=3\theta ~(\bmod ~2\pi)$ and $f$ can have bistability for fixed parameters. We choose
\begin{equation}
f(x,\theta,\beta,\sigma)= \alpha\atan(x) + \beta + \sigma \cos \theta. \label{eq:arctan}
\end{equation}
For the numerical explorations we fix $\alpha=2$ and vary the real parameters $\beta,\sigma$. Most of the statements go through for any $\alpha>1$ though we note that, for example, the regularity of the invariant graph that separates the attractor basins will decrease for $\alpha$ close to $1$.

Note that the map \eqref{eq:map}--\eqref{eq:arctan} is monotonic in $x$: that is, if $(x_n,\theta_n)$ and $(y_n,\theta_n)$ are two forward trajectories with the same $\theta$-component, then
\begin{equation}
\label{eq:monotonicity}
x_n<y_n  ~~\Rightarrow~~ x_{n+1}<y_{n+1}.
\end{equation}
Moreover, all points $x$ throughout the whole real line $\bbR$ are mapped after one iteration into the finite-width interval 
\begin{equation}
\label{eq:boundedness}
f(x,\theta,\beta,\sigma)\in \rI_{\beta,\sigma} :=\left( \beta-\alpha \frac{\pi}{2}-\sigma,\beta+\alpha \frac{\pi}{2}+\sigma\right).
\end{equation}

For the $\theta$-component of the state space of the map \eqref{eq:map}, we regard $[0,2\pi)$ topologically as wrapped round as a circle. A non-empty compact subset of the state space $\bbR\times [0,2\pi)$ is called \emph{invariant} if it is mapped surjectively onto itself by \eqref{eq:map}. The \emph{basin} (\emph{of attraction}) of a compact invariant set $\cA$ is the set of initial conditions $(x_0,\theta_0)$ for which the distance of the trajectory $(x_n,\theta_n)$ from $\cA$ tends to $0$ as $n \to \infty$. A set of initial conditions $U \subset \bbR\times [0,2\pi)$ is said to be \emph{uniformly attracted} to a compact invariant set $\cA$ if there exists a sequence of values $\varphi_n \geq 0$ tending to $0$ such that for every initial condition $(x_0,\theta_0) \in U$, the distance of $(x_n,\theta_n)$ from $\cA$ is at most $\varphi_n$ for all $n$. A (\emph{local}) \emph{attractor} is a compact invariant set $\cA$ admitting a neighbourhood $U$ that is uniformly attracted to $\cA$. An attractor's basin of attraction is always an open set. A \emph{global attractor} is an attractor whose basin is the whole of $\bbR \times [0,2\pi)$.

Observe that the map (\ref{eq:map}) is (assuming periodicity in $\theta$) a smooth skew product map over the chaotic map $g$ for which Lebesgue measure is ergodic. 

The special case $\beta=\sigma=0$ corresponds to a product map with two attractors 
$$
\cA^{\pm}=\{\pm a_0\}\times [0,2\pi)
$$
where $a_0$ is the positive solution of $x=\alpha \atan x$; for $\alpha=2$, we have $a_0 \approx 2.33112$. The attractors share a basin boundary, namely
$$
\cC=\{0\}\times [0,2\pi),
$$
which is an edge state of the system (i.e. it is an attractor for the system restricted to the basin boundary). Note that $\beta\neq 0$ breaks the symmetry of the system $x\mapsto -x$ while $\sigma>0$ will provide bounded chaotic forcing to the dynamics of $x$.  For the noise-unforced case $\sigma=0$ we can say more -- that (\ref{eq:map}) reduces to
\begin{equation}
\label{eq:unforced}
x_{n+1}=\alpha \arctan(x_n)+\beta
\end{equation}
which has a double-fold hysteresis of attractors on varying $\beta$, as shown in Figure~\ref{fig:xmap}. Namely, letting $x^*$ be the unique positive value at which $\left.\frac{\partial f}{\partial x}\right|_{x=x^*}=\frac{\alpha}{1+(x^*)^2}=1$, there is\footnote{Note that $x^*=\sqrt{\alpha-1}$ and so $\beta^*=\alpha \arctan(\sqrt{\alpha-1})-\sqrt{\alpha-1}$. In the particular case $\alpha=2$ we have $x^*= 1$ and $\beta^*=\frac{\pi}{2}-1 \approx 0.5708$.} a $\beta^*>0$ and continuously $\beta$-parameterised real numbers $(a^+_\beta)_{\beta \in [-\beta^*,\infty)}$, $(a^-_\beta)_{\beta \in (-\infty,\beta^*]}$ and $(r_\beta)_{\beta \in [-\beta^*,\beta^*]}$, with $a^\pm_\beta$ strictly increasing in $\beta$ and $r_\beta$ strictly decreasing in $\beta$, such that the following hold for $\sigma=0$:
\begin{enumerate}
    \item For $-\beta^*<\beta<\beta^*$ there are two attractors $\cA^{\pm}_{\beta}=\{a^\pm_\beta\} \times [0,2\pi)$, and these two attractors share the same basin boundary, namely the edge state $\cC_{\beta}=\{r_\beta\} \times [0,2\pi)$. 
    \item The attractor $\cA^+$ and edge state collide at a saddle-node in the $x$ variable at $-\beta^*$, where we have $a^+_{-\beta^*}=r_{-\beta^\ast}=x^*$. For $\beta<-\beta^*$, $\cA^-_\beta=\{a^-_\beta\} \times [0,2\pi)$ is a global attractor.
    \item The attractor $\cA^-$ and edge state  collide at a saddle-node in the $x$-variable at $\beta^*$, where we have $a^-_{\beta^*}=r_{\beta^*}=-x^*$. For $\beta>\beta^*$, $\cA^+_\beta=\{a^+_\beta\} \times [0,2\pi)$ is a global attractor.
\end{enumerate}

\begin{figure}
    \centering
    \includegraphics[width=9cm]{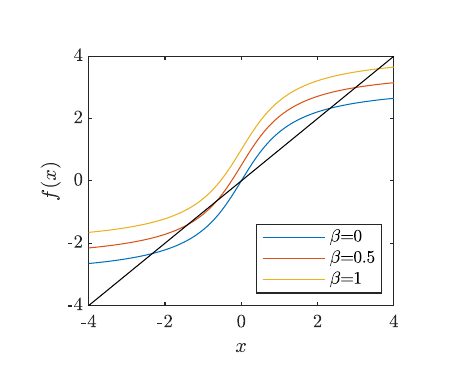}
    \caption{Plots of $x$ vs $f(x)$ for $\alpha=2$ and $\sigma=0$, varying $\beta$. There are two stable fixed points in cases where $|\beta|<\beta^*$ but only one stable fixed point for $|\beta|>\beta^*$. 
    }
    \label{fig:xmap}
\end{figure}

\begin{figure}
    \centering
    (a)\includegraphics[width=8cm]{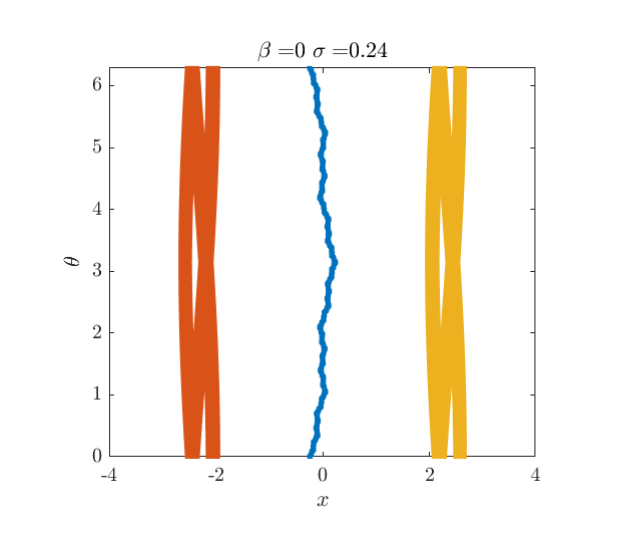}~ (b)\includegraphics[width=8cm]{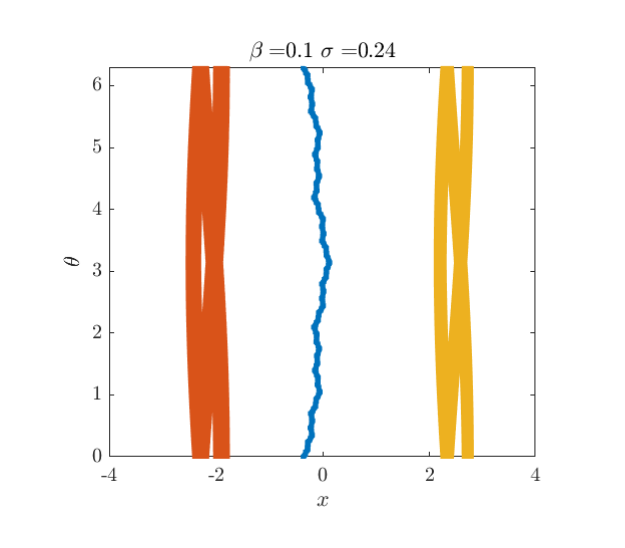}

    (c)\includegraphics[width=8cm]{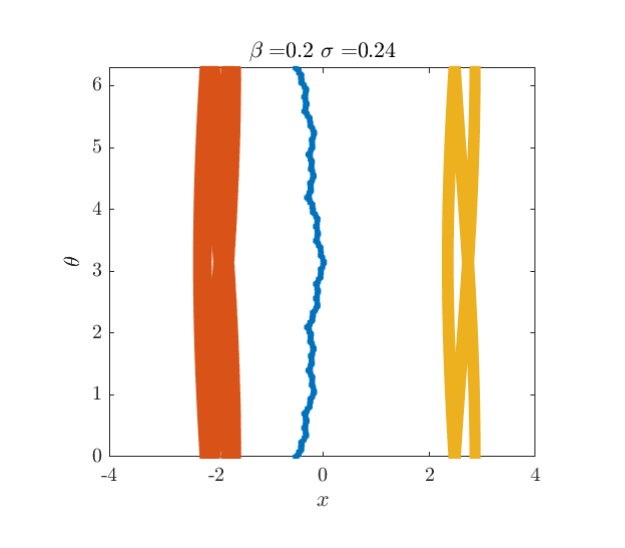}~
    (d)\includegraphics[width=8cm]{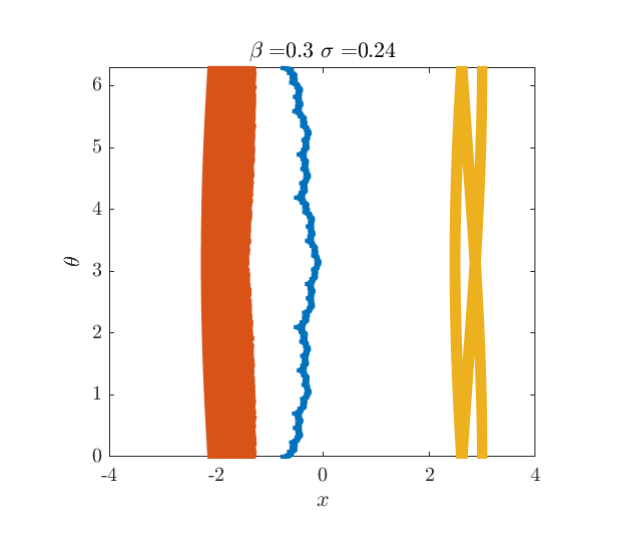}

   (e)\includegraphics[width=8cm]{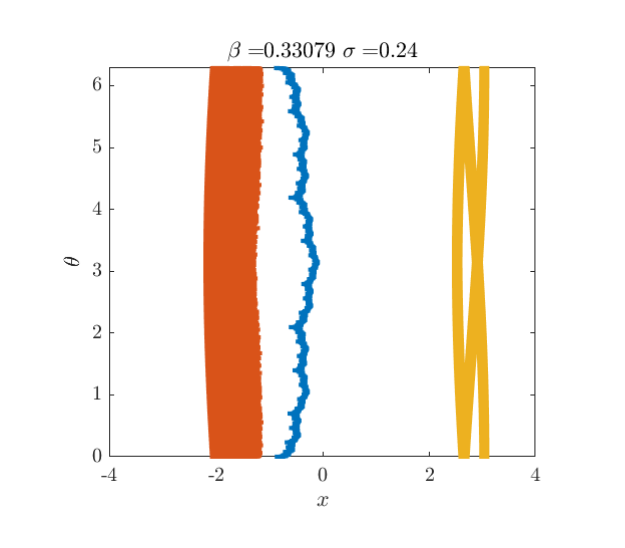}~
    (f)\includegraphics[width=8cm]{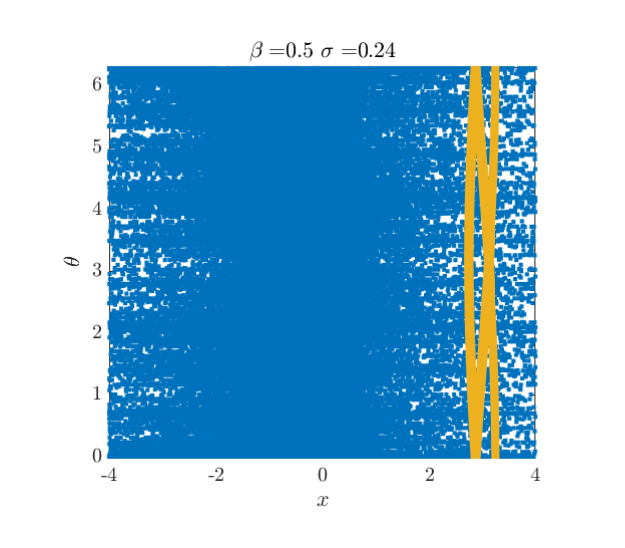}

    \caption{Phase portraits showing numerical approximations of $\cA^{+}$ (yellow), $\cA^-$ (red) and $\cC$ (blue) for $\sigma=0.24$ and (a) $\beta=0.0$ (b) $\beta=0.1$ (c) $\beta=0.2$ (d) $\beta=0.3$ (e) $\beta=0.33079$ (f) $\beta=0.5$. In case (f) the numerically approximated blue  ``saddle'' are in fact only a transient - forward orbits end up at $\cA^+$ if $\beta>\beta^*-\sigma=0.3307$. There is a non-autonomous saddle node of the forced system at $\beta=\pi/2-1\approx 0.330796$.}
    \label{fig:autonomous}
\end{figure}

This picture persists in many regards for $\sigma \neq 0$, as we describe in the next section.

\subsection{Parameter dependence for the map} \label{sec:bulletpoints}

Due to the skew-product structure, the set of tangent vectors parallel to the $x$-axis in the $(x,\theta)$-state space is an invariant subbundle of the tangent bundle for each $\beta$ and $\sigma$. Accordingly, for any $\beta$ and $\sigma$, for any ergodic invariant probability measure $\mu$ on $\mathbb{R} \times [0,2\pi)$, we write $\lambda_{\beta,\sigma}(\mu)$ for the Lyapunov exponent of $\mu$ in this invariant subbundle. The basin of attraction of an invariant probability measure $\mu$ refers to the set of initial conditions $(x_0,\theta_0)$ for which the empirical measure $\frac{1}{N} \sum_{n=0}^{N-1} \delta_{(x_n,\theta_n)}$ converges weakly to $\mu$ as $N \to \infty$; if the basin of attraction of $\mu$ has positive Lebesgue measure, then we say that $\mu$ is a \emph{physical measure} (see e.g.\ \cite{Young:2016}).

We consider how the dynamics varies qualitatively over the $(\beta,\sigma)$-parameter space $\bbR \times [0,\infty)$ shown in Figure~\ref{fig:parameter_space}. Note that in this parameter space, a line of gradient $+1$ corresponds to a constant value of $\beta-\sigma$ and a line of gradient $-1$ corresponds to a constant value of $\beta+\sigma$ (with these constant values being equal to the $\beta$-value at the point where the line intersects $\{\sigma=0\}$). We define the following regions of the parameter space, as depicted in Figure~\ref{fig:parameter_space}:
$$
\begin{array}{rl}
    S^- &= \{(\beta,\sigma) \in \bbR \times [0,\infty) : \beta+\sigma \leq \beta^* \} \\
    S^+ &= \{(\beta,\sigma) \in \bbR \times [0,\infty) : \beta-\sigma \geq -\beta^* \} \\
    \rA &= S^+ \cap S^- = \mathrm{hull}\Big( (-\beta^*,0) \, , \, (0,\beta^*) \, , \, (\beta^*,0) \Big) \\
    \rB^- &= \{(\beta,\sigma) \in \bbR \times [0,\infty) : \beta+\sigma < -\beta^* \} \subset S^- \setminus S^+ \\ 
    \rB^+ &= \{(\beta,\sigma) \in \bbR \times [0,\infty) : \beta-\sigma > \beta^* \} \subset S^+ \setminus S^- \\
    \rC^- &= S^- \setminus (S^+ \cup \rB^-) \\ 
    \rC^+ &= S^+ \setminus (S^- \cup \rB^+) \\
    \rD &= \big(\bbR \times [0,\infty)\big) \setminus \big( S^+ \cup S^- \big).
\end{array}
$$
By construction, the six regions $\rA$, $\rB^-$, $\rC^-$, $\rD$, $\rC^+$, $\rB^+$ (as seen going clockwise through Figure~\ref{fig:parameter_space}) are mutually disjoint and form a partition of the whole parameter space $\bbR \times [0,\infty)$. Note that the pairs $S^{\pm}$, $\rB^{\pm}$ and $\rC^{\pm}$ are respectively each others reflection in the $\sigma$-axis. The sets $\rA$ and $\rD$ are both symmetric in the $\sigma$-axis.

Let us also define the closed line-segments:
\[
L^- = \{(\beta,\beta^*+\beta) : \beta \in [-\beta^*,0] \} 
\qquad 
L^+ = \{(\beta,\beta^*-\beta) : \beta \in [0,\beta^*] \}. \]
In other words, $L^-$ is the intersection of the boundaries of $\rA$ and $\rC^-$; and $L^+$ is the intersection of the boundaries of $\rA$ and $\rC^+$. So the boundary of $\rA$ (relative to $\bbR \times [0,\infty)$) is precisely $L^- \cup L^+$.
Note that $S^{\pm}$ and $\rA$ are closed sets; and $\rB^{\pm}$ and $\rD$ are (relative to $\bbR \times [0,\infty)$) open sets. The sets $\rC^{\pm}$ are neither open nor closed: they exclude the side of their boundary constituted by $L^-$ and $L^+$ respectively, but they each include the open half-lines forming the remaining two sides of their three-sided boundary. Proofs of the statements below (except those regarding region~$\rD$) are given in Appendix~\ref{app:proofs}.

\subsubsection*{Attractors and physical measures in $S^\pm$}

In each of the sentences in the bullet points below, two statements are being made: one where $\pm$ and $\mp$ are consistently substituted for $+$ and $-$ respectively throughout, and one where $\pm$ and $\mp$ are consistently substituted for $-$ and $+$ respectively throughout.
\begin{itemize}
    \item For all $(\beta,\sigma) \in S^\pm$, there is a unique (so in particular, ergodic) invariant probability measure $\mu_{\beta,\sigma}^\pm$ whose support $\cA_{\beta,\sigma}^\pm$ is contained in $[a^\pm_{\beta-\sigma},a^\pm_{\beta+\sigma}] \times [0,2\pi)$ and whose $\theta$-marginal is the normalised Lebesgue measure on $[0,2\pi)$.
    
    (See Secs.~\ref{sec:A1}, \ref{sec:App_ergcts}, \ref{sec:App_thetaspan} and the end of \ref{sec:App_basinmu} for statements about $\mu^+$; analogous statements for $\mu^-$ hold by Sec.~\ref{sec:App_mu-}.)
    \item For $(\beta,\sigma)$ in the interior of $\rS^\pm$, $\cA_{\beta,\sigma}^\pm$ is an attractor.

    (See Corollary~\ref{cor:App_localattr} in Sec.~\ref{sec:App_attrA+}.)
    \item For $(\beta,\sigma) \in S^\pm$, Lebesgue-almost every point in the basin of $\cA_{\beta,\sigma}^\pm$ is in the basin of $\mu_{\beta,\sigma}^\pm$; in particular, $\mu_{\beta,\sigma}^\pm$ is a physical measure.

    (See Proposition~\ref{prop:App_phys} in Sec.~\ref{sec:App_basinmu}.)
    \item For $(\beta,\sigma) \in S^\pm$ with $\sigma>0$, the intersection of $\cA_{\beta,\sigma}^\pm$ with the line $\cA_{\beta-\sigma}^\pm$ is the fixed point $(a^\pm_{\beta-\sigma},\pi)$ and the intersection of $\cA_{\beta,\sigma}^\pm$ with the line $\cA_{\beta+\sigma}^\pm$ is the fixed point $(a^\pm_{\beta+\sigma},0)$.\footnote{Since $\cA_{\beta,\sigma}^\pm$ is an invariant set, and $0$ and $\pi$ are fixed points of $g$, the $(\theta=0)$-section and the $(\theta=\pi)$-section of $\cA_{\beta,\sigma}^\pm$ must be \emph{positively} invariant sets (i.e.\ they map into themselves, but do not have to map onto the whole of themselves); one can show that they contract in forward time towards the singletons $\{(a^\pm_{\beta+\sigma},0)\}$ and $\{(a^\pm_{\beta-\sigma},\pi)\}$ respectively, but that if $\sigma>0$ then they include more than just those singletons and therefore are not invariant sets.}

    (See Sec.~\ref{sec:App_thetaspan}.)
    \item For $(\beta,\sigma) \in S^+ \setminus \{(\mp \beta^*,0)\}$, the Lyapunov exponent $\lambda_{\beta,\sigma}(\mu_{\beta,\sigma}^\pm)$ is strictly negative.

    (See Sec.~\ref{sec:App_LE}.)
    \item Over $(\beta,\sigma) \in S^\pm$: the dependence of $\mu_{\beta,\sigma}^\pm$ on $(\beta,\sigma)$ is continuous in the topology of weak convergence; the dependence of $\cA_{\beta,\sigma}^\pm$ on $(\beta,\sigma)$ is continuous in Hausdorff distance; and the dependence of $\lambda_{\beta,\sigma}(\mu_{\beta,\sigma}^\pm)$ on $(\beta,\sigma)$ is continuous.

    (See Sec.~\ref{sec:App_ergcts}.)
\end{itemize}

\subsubsection*{Basin-boundary dynamics in $\rA$}

\begin{itemize}
    \item For $(\beta,\sigma) \in \rA$, the basin of attraction of $\cA_{\beta,\sigma}^-$ and the basin of attraction of $\cA_{\beta,\sigma}^+$ share the same boundary $\cC_{\beta,\sigma}$. This basin boundary is contained in $[r_{\beta+\sigma},r_{\beta-\sigma}] \times [0,2\pi)$.

    (See Sec.~\ref{sec:App_BB}, which is contingent on definitions in Sec.~\ref{sec:App_rconstr}.)
    \item There exists a continuous map $(\beta,\sigma,\theta) \mapsto r_{\beta,\sigma}(\theta)$ from $\rA \times [0,2\pi)$ to $[-x^\ast,x^\ast]$ such that for each $(\beta,\sigma) \in \rA$, the basin boundary $\cC_{\beta,\sigma}$ is precisely the graph $\{(r_{\beta,\sigma}(\theta),\theta) : \theta \in [0,2\pi)\}$ of $r_{\beta,\sigma}$.

    (See Secs.~\ref{sec:App_rconstr}, \ref{sec:App_prop_r} and \ref{sec:App_BB}.)
    \item For $(\beta,\sigma) \in \rA$ with $\sigma>0$, $r_{\beta,\sigma}(\cdot)$ attains the value $r_{\beta+\sigma}$ uniquely at $\theta=0$, and $r_{\beta,\sigma}(\cdot)$ attains the value $r_{\beta-\sigma}$ uniquely at $\theta=\pi$.
    \item For $(\beta,\sigma) \in \rA$, the pushforward $\nu_{\beta,\sigma}$ of the normalised Lebesgue measure under $\theta \mapsto (r_{\beta,\sigma}(\theta),\theta)$ is an ergodic invariant probability measure.

    (See Sec.~\ref{sec:App_prop_r}.)
    \item For $(\beta,\sigma) \in \rA \setminus \{(-\beta^\ast,0),(\beta^\ast,0)\}$, the Lyapunov exponent $\lambda_{\beta,\sigma}(\nu_{\beta,\sigma})$ is strictly positive.

    (See Sec.~\ref{sec:App_prop_r}.)
    \item Due to the continuity of $r_{\beta,\sigma}(\theta)$ over $(\beta,\sigma,\theta) \in \rA \times [0,2\pi)$: the support of $\nu_{\beta,\sigma}$ is $\cC_{\beta,\sigma}$; the dependence of $\nu_{\beta,\sigma}$ on $(\beta,\sigma)$ is continuous in the topology of weak convergence; the dependence of $\cC_{\beta,\sigma}$ on $(\beta,\sigma)$ is continuous in Hausdorff distance; and the dependence of $\lambda_{\beta,\sigma}(\nu_{\beta,\sigma})$ on $(\beta,\sigma)$ is continuous.

    (See Sec.~\ref{sec:App_prop_r}.)
\end{itemize}

\subsubsection*{Intersections of $\cA^\pm_{\beta,\sigma}$ with $\cC_{\beta,\sigma}$}

Combining some of the points made so far:
    \begin{itemize}
        \item For $(\beta,\sigma)$ in the interior of $\rA$, $\cC_{\beta,\sigma}$ has empty intersection with $\cA_{\beta,\sigma}^-$ and with $\cA_{\beta,\sigma}^+$.
        \item For $(\beta,\sigma) \in L^+$, $\cC_{\beta,\sigma}$ has non-empty intersection with $\cA_{\beta,\sigma}^-$; except at $(\beta,\sigma)=(\beta^\ast,0)$, this intersection is the single point $(-x^\ast,0)$.
        \item For $(\beta,\sigma) \in L^-$, $\cC_{\beta,\sigma}$ has non-empty intersection with $\cA_{\beta,\sigma}^+$; except at $(\beta,\sigma)=(-\beta^\ast,0)$, this intersection is the single point $(x^\ast,\pi)$. 
    \end{itemize}
(See Sec.~\ref{sec:App_inters} for all three statements.)

\subsubsection*{Attractor basins in $S^\pm \setminus \rA$}

\begin{itemize}
    \item For $(\beta,\sigma)\in\rC^+$, the basin of $\cA_{\beta,\sigma}^+$ has full Lebesgue measure and includes $(r_{\beta-\sigma},\infty) \times [0,2\pi)$.
    \item For $(\beta,\sigma)\in\rC^-$, the basin of $\cA_{\beta,\sigma}^-$ has full Lebesgue measure and includes $(-\infty,r_{\beta+\sigma}) \times [0,2\pi)$.
    \item For $(\beta,\sigma)\in\rB^{\pm}$, $\cA_{\beta,\sigma}^\pm$ is a global attractor.
\end{itemize}
(See Sec.~\ref{sec:App_Bfull} for all three statements.)

In all three cases, $\mu_{\beta,\sigma}^\pm$ is the only invariant measure whose $\theta$-projection is the normalised Lebesgue measure. (See the end of Sec.~\ref{sec:App_basinmu}.)

\subsubsection*{Dynamics in and near $\rD$}

For each $(\beta,\sigma) \in \rD$:
\begin{itemize}
    \item For every initial condition $(x_0,\theta_0) \in \bbR \times [0,2\pi)$, the distance of the trajectory $(x_n,\theta_n)$ from the set $[a_{\beta-\sigma}^-,a_{\beta+\sigma}^+] \times [0,2\pi)$ tends to $0$ as $n \to \infty$.

    (See Proposition~\ref{prop:App_big_ga} in Sec.~\ref{sec:App_Dinv}.)
    \item Every invariant probability measure has its support contained in $[a_{\beta-\sigma}^-,a_{\beta+\sigma}^+] \times [0,2\pi)$.

    (See Sec.~\ref{sec:App_Dinv}, after Proposition~\ref{prop:App_big_ga}.)
    \item Every invariant measure whose $\theta$-marginal is the normalised Lebesgue measure has the points $(a_{\beta-\sigma}^-,\pi)$ and $(a_{\beta+\sigma}^+,0)$ in its support.

    (See Proposition~\ref{prop:App_endpoints} in Sec.~\ref{sec:App_Dinv}.)
    \item There exists at least one ergodic invariant measure whose $\theta$-marginal is the normalised Lebesgue measure.

    (See Sec.~\ref{sec:App_Dinv}, after Proposition~\ref{prop:App_big_ga}.)
\end{itemize}

Now for all $(\beta,\sigma) \in \bbR \times [0,\infty)$, the following statements are equivalent:
    \begin{itemize}
        \item[(i)] There is only one invariant measure whose $\theta$-marginal is the normalised Lebesgue measure.
        \item[(ii)] For a uniformly distributed $[0,2\pi)$-valued random variable $\theta$, the width $X_{\beta,\sigma,n}$ of the $n$-th image of $\bbR \times \{\theta\}$ under the map~\eqref{eq:map} converges in probability to $0$ as $n \to \infty$.
    \end{itemize}
(See Proposition~\ref{prop:App_tildeD} in Sec.~\ref{sec:App_GC}.) Note that $X_{\beta,\sigma,n}$ is indeed finite for each $n \geq 1$, due to Eqs.~\eqref{eq:monotonicity} and \eqref{eq:boundedness}.

Let $\tilde{\rD}$ be the set of all $(\beta,\sigma) \in \bbR \times [0,\infty)$ for which the equivalent statements~(i) and (ii) above hold. For each $(\beta,\sigma) \in \tilde{\rD}$, we will write $\mu_{\beta,\sigma}$ for the unique invariant measure whose $\theta$-marginal is the normalised Lebesgue measure, and we will write $\cA_{\beta,\sigma}$ for the support of $\mu_{\beta,\sigma}$. Let $\tilde{\rD}_{\lambda<0} \subset \tilde{\rD}$ be the set of all $(\beta,\sigma) \in \tilde{\rD}$ for which $\lambda_{\beta,\sigma}(\mu_{\beta,\sigma})<0$.

We have already established that $\tilde{\rD}_{\lambda<0}$ includes $(S^- \cup S^+) \setminus \rA$, and that $\tilde{\rD}$ has empty intersection with $\rA$. So to determine the sets $\tilde{\rD}$ and $\tilde{\rD}_{\lambda<0}$, all that would remain is to determine their intersection with $\rD$.

\begin{itemize}
    \item Over $(\beta,\sigma) \in \tilde{\rD}$, the dependence of $\mu_{\beta,\sigma}$ on $(\beta,\sigma)$ is continuous in the topology of weak convergence.

    (See Sec.~\ref{sec:App_ctsD}.)
    \item For $(\beta,\sigma) \in \tilde{\rD}_{\lambda<0}$, $X_{\beta,\sigma,n}$ converges exponentially to $0$ almost surely as $n \to \infty$.

    (See Proposition~\ref{prop:App_expconv} in Sec.~\ref{sec:App_GC}.)
    \item Nevertheless, there exist $(\beta,\sigma) \in \rD$ (such as $(0,\sigma)$ for any $\sigma>\beta^*$) for which there is no deterministic sequence $c_n>0$ tending to $0$ such that $X_{\beta,\sigma,n}$ is almost surely less than $c_n$.

    (See Proposition~\ref{prop:App_pi/2} in Sec.~\ref{sec:App_GC}.)
    \item For $(\beta,\sigma) \in \tilde{\rD}_{\lambda<0}$, Lebesgue-almost every point in $\bbR \times [0,2\pi)$ is in the basin of attraction of $\cA_{\beta,\sigma}$ and in the basin of attraction of $\mu_{\beta,\sigma}$; so in particular, $\mu_{\beta,\sigma}$ is a physical measure.

    (See Corollary~\ref{cor:App_Dbasin} and Proposition~\ref{prop:App_physD} in Sec.~\ref{sec:App_BasinD}.)
\end{itemize}

We conjecture that $\tilde{\rD}_{\lambda<0}$ contains the whole of $\rD$; this is equivalent to saying that
\[ \tilde{\rD} = \tilde{\rD}_{\lambda<0} = (\bbR \times [0,\infty)) \setminus \rA. \]
If this conjecture is true, then the physical measure $\mu_{\beta,\sigma}$ has continuous dependence in the topology of weak convergence over $(\beta,\sigma)$ varying throughout the complement of $\rA$. We note, however, that while the dependence in Hausdorff distance of $\cA_{\beta,\sigma}$ is continuous over $(\beta,\sigma) \in (S^- \cup S^+) \setminus \rA$, it cannot extend continuously from $(S^- \cup S^+) \setminus \rA$ into $\rD$. Numerical support for this conjecture is presented in Figure~\ref{fig:parameter_numerics}. Note that by Eqs.~\eqref{eq:monotonicity} and \eqref{eq:boundedness}, given any finite-width interval $[x_1^-,x_1^+] \supset \rI_{\beta,\sigma}$, any forward trajectory $(x_n,\theta_n)$ of \eqref{eq:map} starting at $n=0$ has its $x$-component lying between the ``bounding trajectories''
\begin{equation} \label{eq:attracting}
    x_n^- \leq x_n \leq x_n^+ \quad \forall \, n \geq 1
\end{equation}
where $(x_n^\pm,\theta_n)$ are forward trajectories starting at $n=1$ with the pre-fixed $x$-component $x_1^\pm$. So then, the random variable $X_{\beta,\sigma,n}$ is bounded above by the width of the $(n-1)$-th image of $[x_1^-,x_1^+] \times \{g(\theta)\}$ under the map~\eqref{eq:map}. So since the normalised Lebesgue measure on $[0,2\pi)$ is $g$-invariant, if we fix any finite-width interval $[x_0^-,x_0^+] \supset \rI_{\beta,\sigma}$, we have that $(\beta,\sigma) \in \tilde{\rD}$ if and only if the width of the $n$-th image of $[x_0^-,x_0^+] \times \{\theta\}$ under the map~\eqref{eq:map} converges in probability to $0$ as $n \to \infty$. So in the left plot of Figure~\ref{fig:parameter_numerics}, fixing such values $x_0^-,x_0^+$ and selecting a value $\theta_0$ from the uniform distribution on $[0,2\pi)$, we simulate subsequent trajectories $(x_n^\pm,\theta_n)$ starting at $(x_0^\pm,\theta_0)$ and plot over $(\beta,\sigma)$-space the value of $x_n^+-x_n^-$ for a large value of $n$. We see values close to $0$ everywhere except in and near $\rA$, supporting the conjecture that $\tilde{\rD}=(\bbR \times [0,\infty)) \setminus \rA$. In the right plot of Figure~\ref{fig:parameter_numerics}, we plot the fibre LE for the attractor (or the mximum of the fibre LE for the two attractors) found in the left plot. We see strictly negative values throughout the whole $(\beta,\sigma)$-parameter space, supporting the conjecture that $\tilde{\rD}_{\lambda<0}$ is the whole of $\tilde{\rD}$.

\begin{figure}
    \centering

    \includegraphics[width=18cm]{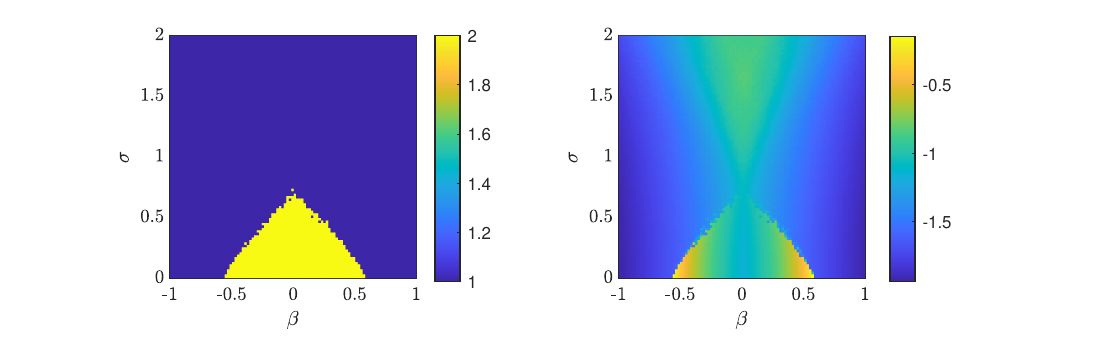}
    
    \caption{Left: Parameter space $(\beta,\sigma)$ showing number of attractors in each region, where blue indicates monostability. Bistability occurs in the yellow triangular region (see region $\rA$ in Figure~\ref{fig:parameter_space}. Right: Largest fibre Lyapunov exponent (LE) for any attractor shown for each point in parameter space, computed for a randomly chosen $\theta_0$. Note the non-autonomous saddle-node bifurcations on the boundary of the region of bistability.
    }
    \label{fig:parameter_numerics}
\end{figure}

\subsubsection*{Non-autonomous saddle-node bifurcation}

The boundary $L^\pm$ between regions~$\rA$ and $\rC^{\pm}$ can be seen as a \emph{non-autonomous saddle-node bifurcation}~\cite{Anagnostopoulou:2012} where one of the attractors $\cA^\pm_{\beta,\sigma}$ collides with the repeller $\cC_{\beta,\sigma}$, albeit just at a single fixed point. Let us emphasise that in contrast to autonomous saddle-node bifurcations (where $\sigma=0$), for $\sigma>0$ the Lyapunov exponent $\lambda^\pm_{\beta,\sigma}$ is strictly negative even at the critical boundary in the $(\beta,\sigma)$-parameter space. 

\subsection{Trajectories and UPOs} \label{sec:UPO}

Figure~\ref{fig:autonomous} illustrates typical dynamical behaviours one can find for the map \eqref{eq:map}--\eqref{eq:arctan}. The yellow and red trajectories are attractors found by iterating forwards respectively from $(\pm 2.3,0.1)$ and plotting $N=10^6$ points after omitting the first $T=50$ points as transient: they approximate $\cA^{\pm}_{\beta,\sigma}$. The blue trajectory is found by backwards iteration of the $x$ dynamics: this approximates the saddle/repeller $\cC_{\beta,\sigma}$. A forwards orbit $\theta_n=g^n(\theta_0)$ is generated then we choose $x_{N+T}=0$ and then define
$$
x_{n-1}= h(x_n,\theta_{n-1})
$$
where $h(x,\theta)= \tan((x-\beta-\sigma \cos \theta)/\alpha)$: note that this is the inverse of $f$ in the sense that
$$
h(f(x,\theta),\theta)=x.
$$
The trajectory $(x_i,\theta_i)$ is plotted for $i=1,\ldots,N$; the saddle-node bifurcation lines forced by unstable periodic orbits (UPOs) in Figure~\ref{fig:parameter_space}. 

We construct UPOs of period $N$ for $g$ symbolically by choosing a sequence $a_i\in \{0,1,2\}$ for $i=1..N$ and noting that the unique periodic point with coding $\{a_i\}_{i=1}^N$ is
$$
\theta=2\pi\sum_{i=1}^{N}a_i \frac{3^{N-i}}{3^N-1}.
$$
We force with the orbit of this point to numerically determine the saddle-node bifurcations for the return map forced over one period of the UPO; this is shown in Figure~\ref{fig:parameter_space} for all orbits of period $N\leq 3$ and selected orbits with $N=4$.

\begin{figure}
    \centering
    \includegraphics[width=12cm]{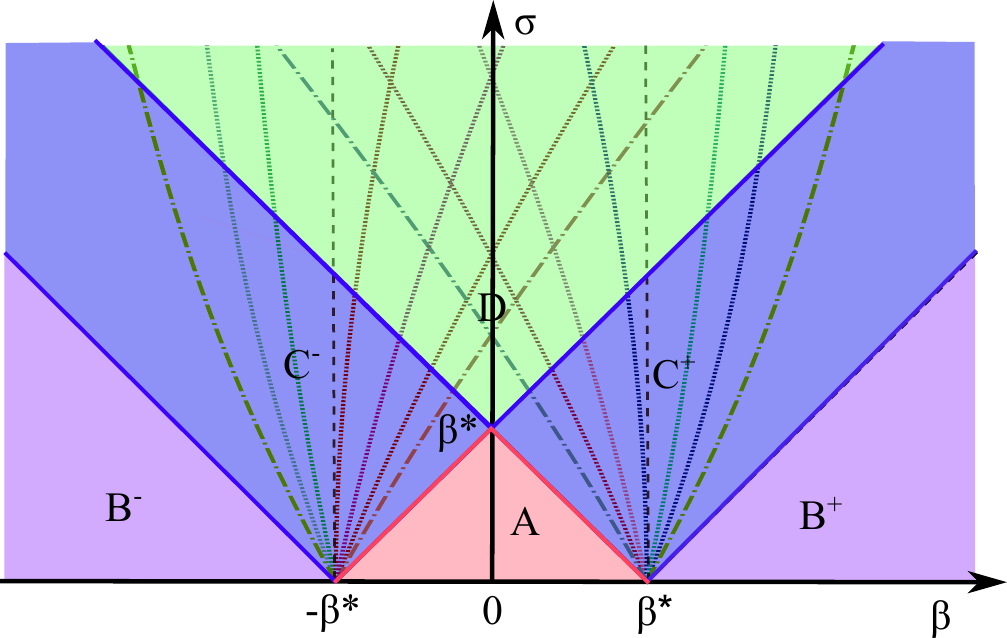}

    \caption{Division of the $(\beta,\sigma)$-parameter space $\mathbb{R} \times [0,\infty)$ for (\ref{eq:map}) and $\alpha=2$, into regions exhibiting qualitatively different dynamics. We discuss regions $S^-=\rB^-\cup\rC^-\cup \rA$ and  $S^+=\rB^+\cup\rC^+\cup \rA$ in the text. Note that the system is bistable in region $\rA$ and monostable elsewhere (see Figure~\ref{fig:parameter_numerics}). For parameters in regions $\rC^{\pm}$ the attractor is localised to one of the branches of attractors for $\sigma=0$ but there is non-uniform approach to the attractor. In region $\rD$ the attractors have fused into one large attractor. The solid/dashed/dotted/dash-dotted lines correspond to saddle-node bifurcations forced respectively by unstable fixed points/period 2 points/period 3 points/selected period 4 points of the $\theta$ dynamics. 
    }
    \label{fig:parameter_space}
\end{figure}

\subsection{Geometry and dimension of the chaotic saddle}

In cases of bistability, the basin boundary contains a chaotic saddle $\cC_{\beta,\sigma}$ (actually a repeller) that occupies all of the boundary. As observed in Figure~\ref{fig:autonomous}, this invariant set is a repeller of the skew product system whose geometry we remark on here in relation to the literature of repellers in skew affine systems; see for example \cite{walkden2018invariant,shen2018hausdorff}. In this section we give some numerical estimates of the dimension of $\cC_{\beta,\sigma}$ both via a box counting algorithm  \cite{falconer2007fractal} and  using the Lyapunov exponents and escape time from $\cC_{\beta,\sigma}$  to give a Lyapunov dimension \cite{Hunt1996}, as estimates of the Hausdorff dimension of $\cC_{\beta,\sigma}$.

Note that for the affine skew product
$$
x_{n+1}=[\lambda(\theta_n)]^{-1} x_n+\sigma \cos(\theta_n)
$$
with $\lambda(\theta)>1$ there will be a repelling invariant graph \cite{bedford1989} whose box dimension is given in terms of the thermodynamic formalism by  \cite{bedford1989} as a $t$ that is a zero of a pressure function that in our case is
$$
\cP\left((1-t)\ln |\partial_{\theta}g|-\ln |\partial_x f|\right)=0
$$
also called Bowen's formula \cite{bedford1989,walkden2018invariant}. In the case that the expansion in the $\theta$ direction is uniform at rate $B>1$ and in the $x$ direction is uniform at rate $A>1$ this reduces to
$$
\ln B+ (1-t)\ln B -\ln A=0
$$
and so it is possible to show that the invariant graph has Hausdorff dimension $2-\ln(A)/\ln(B)$. These results do not obviously extend to nonlinear skew products so we explore the dimensions of the chaotic saddle numerically.

\subsubsection{Box-counting dimension} \label{sbsec:Box-Counting}

Making use of the skew product structure of the system we approximate the shape of the saddle via forward and backward iterations as follows: We choose a box $B_{\beta,\sigma} = [ x_l,x_u ] \times [ \theta_l, \theta_u]$ that contains a part of $\cC_{\beta,\sigma}$, and randomly choose $k$ values of $\theta \in [ \theta_l, \theta_u]$. Iterating the base map (i.e. $\theta_{n+1}= 3\theta_n \mod 2$) forward for $N$ iterations and storing the trajectories, allows evaluation of the inverted fibre map under $N$ iterations backward in time. Note that under time inversion, the fibre map becomes attracting towards $\cC_{\beta, \sigma}$.
This procedure results in an approximation of the saddle with $k$ points.

Then, we compute the saddle's box-counting dimension. This is done by considering the number $T(\kappa)$ of non-intersecting boxes with edge-length $\kappa$ that is needed to cover $\cC_{\beta,\sigma}$. 
The dimension is then given by:
\begin{align}
    d_{box}(\cC_{\beta,\sigma})&= \lim_{\kappa \to 0} \frac {\ln T(\kappa)}{\ln(1/\kappa)}
\end{align}
Rearranging this formula for a finite value $\kappa$ gives 
\begin{align}\label{eq:box_dim_lin_reg}
    \ln T(\kappa) & \propto d_{box}(\cC)\cdot \ln T(\kappa). 
\end{align}
Thus, plotting $\ln T(\kappa)$ vs. $\ln T(\kappa)$, and assuming that this relation correctly describes the behaviour for small $\kappa$ allows an estimation of the box counting dimension via a linear fit\cite{Russel1980}. 
The results are presented in Table~\ref{tab:bc-Lyap-dim}.

\subsubsection{Lyapunov dimension}

Using results from \cite{Hunt1996}, we can calculate the Lyapunov dimension of $\cC_{\beta,\sigma}$ from the Lyapunov exponents as
\begin{align}
	D_{L} &= J + \frac{H - (h_1^+ + \cdot \cdot \cdot h_J^+)}{h_{J+1}^+},
\end{align} 
with $h_1^+ \leq h_2^+$ the two (positive) Lyapunov exponents, and $J=1$ given by
\begin{align}
	h_1^+ + \cdot \cdot \cdot h_J^+ +  h_{J+1}^+ \geq H \geq h_1^+ + \cdot \cdot \cdot h_J^+ .
\end{align}
$H$ is defined as  
\begin{align}
	H &= \sum_{i=1}^2 h_i^+ - \tau^{-1},
\end{align} 
and $\tau$ is the mean escape time from $\cC_{\beta,\sigma}$. This is defined by $\tau ^{-1} = - \lim_{n \rightarrow \infty} \frac{1}{n} \ln \left( \frac{\ell(R^{(n)})}{\ell(R)} \right)$ where $R$ is a box that contains $\cC_{\beta,\sigma}$ but no other invariant sets. $\ell(R)$ is the Lebesgue measure of the set $R$, and $\ell(R^{(n)})$ the Lebesgue measure of the set of trajectories that are still in $R$ after $n$ iterations. A numerical approximation can be done with a linear fit analogously to the procedure described after equation \ref{eq:box_dim_lin_reg} for the box-counting dimension.

To compute the Lyapunov exponents of $\cC_{\beta,\sigma}$, we make use of the skew product structure of the system (\ref{eq:map}) and Birkhoff's ergodic theorem. These allow us to compute the Lyapunov exponents as ensemble averages over the logarithm of the directional derivatives of sample points drawn from the invariant measure on the saddle $\cC_{\beta,\sigma}$. We approximate this invariant measure via the forward and backward evolution algorithm as described in Section~\ref{sbsec:Box-Counting} to compute the Lyapunov exponents: 
\begin{align}
    \lambda_1&=\left\langle \ln \bigg|\frac{\partial f(x, \theta, \beta, \sigma)}{\partial x}\bigg| \right\rangle \\
    \lambda_2&= \ln (3) = 1.0986
\end{align}
where $\langle\cdot\rangle$ denotes the average over an ensemble of $10^4$ points chosen randomly with respect to the invariant measure on the saddle. Using these Lyapunov exponents gives the Lyapunov dimensions shown in Table~\ref{tab:bc-Lyap-dim}. The parameters used for the sampling of $\cC_{\beta,\sigma}$ are shown in Table~\ref{tab:bc-Lyap-dim-parameters}. We note that the numerically estimated dimensions are limited by various factors, in particular the sampling of $\cC_{\beta,\sigma}$ may be quite sparse in places.

\begin{table}
    \centering
    \begin{tabular}{c|c|c|c|c|c}
        $\sigma$  &  $\beta$   &  $\lambda_1(\cC_{\beta,\sigma})$  &  $\tau^{-1}$ &  $d_{box}(\cC_{\beta,\sigma})$  &  $d_{L}(\cC_{\beta,\sigma})$ \\ \hline
        0.24  &  0.0  &  0.684  &  0.684 &  1.326  &  1.378   \\
        0.24  &  0.1  &  0.673  &  0.673 &  1.336  &  1.388   \\
        0.24  &  0.2  &  0.640  &  0.636 &  1.360  &  1.421   \\
        0.24  &  0.3  &  0.580  &  0.572 &  1.429  &  1.480  
    \end{tabular} 
    \caption{Numerically estimated fibre Lyapunov exponent $\lambda_1$, mean escape time $\tau$, box counting dimension $d_{box}$ and Lyapunov dimensions $d_{L}$ of $\cC_{\beta,\sigma}$ for a range of $\beta$ and $\sigma=0.24$ such that there is bistability. Observe the monotonic increase in both dimension estimates on increasing $\beta+\sigma$ towards to the non-autonomous saddle-node bifurcation at $\beta^*$.}
    \label{tab:bc-Lyap-dim}
\end{table}

\begin{table}
    \centering
    \begin{tabular}{c|c|c|c|c|c|c}
        $\sigma$  &  $\beta$   &  $n_\theta$ & $B_{[ x_l,x_u ]}$ & $B_{[ \theta_l, \theta_u]}$ & $R_{[ x_l,x_u ]}$ & $R_{[ \theta_l, \theta_u]}$ \\ \hline
        0.24  &  0.0  &  476  &  [0.155, 0.176]  &  [2.925, 2.95]  &  [-0.3, 0.3]  &  [0.0, $2\pi$]  \\
        0.24  &  0.1  &  441  &  [0.0523, 0.075]  &  [2.925, 2.95]  &  [-0.4, 0.2]  &  [0.0, $2\pi$]  \\
        0.24  &  0.2  &  333  &  [-0.055, -0.025]  &  [2.925, 2.95]  &  [-0.55, 0.05]  &  [0.0, $2\pi$]  \\
        0.24  &  0.3  &  200  &  [-0.175, -0.125]  &  [2.925, 2.95]  &  [-0.75, -0.05]  &  [0.0, $2\pi$]
    \end{tabular} 
    \caption{Parameters chosen to compute the box counting dimensions of $\cC_{\beta,\sigma}$ for different values of $\sigma$ and $\beta$. $n_\theta$ is the number of grid points in $\theta-$direction (resolution of the saddle in $\theta$-direction), $B_{[ x_l,x_u ]}$ and $B_{[ \theta_l, \theta_u]}$ give the size of the box $B_{\beta,\sigma}$ in which we approximated the saddle, and $R_{[ x_l,x_u ]}$ and $R_{[ \theta_l, \theta_u]}$ give the size of the box $R_{\beta,\sigma}$ that contains the saddle but no other invariant set. 
    The following parameters are the same for all values of $\beta$ and $\sigma$: the number of considered orbits is $k=10^6$, the time of forward and backward iterations is $N=10^3$, and the number of grid points in $x$-direction (resolution of the saddle in $x$-direction) is $n_x=400$.}
    \label{tab:bc-Lyap-dim-parameters}
\end{table}

\section{Slowly varying parameters and the dynamic tipping window}
\label{sec:ramped}

Returning to a discrete version of (\ref{xaut2}) where we allow the parameters to slowly vary in addition to chaotic forcing, one can study the interaction between bifurcation and noise induced effects (B- and N-tipping in the terminology of \cite{Ashwinetal:2012}). We investigate this by ramping the parameter beta
\begin{equation}
\beta_n=\beta_0+n\epsilon
\end{equation}
for small $\epsilon$ and simulate a non-autonomous version of (\ref{eq:map}), namely:
\begin{equation}
\label{eq:namap}
\begin{aligned}
x_{n+1} = & f(x_n,\theta_n,\beta_n)\\
\theta_{n+1}=  &g(\theta_n).
\end{aligned}
\end{equation}
We give an initial condition with negative $x_0$, typical $\theta_0$ and $\beta_0=0$ and fix $\alpha=2$ as above and find that the chaotic tipping window at $\epsilon=0$ is the limit of a {\em dynamic tipping window} for $\epsilon>0$.

Figure~\ref{fig:ramp} illustrates trajectories of this non-autonomous map (\ref{eq:namap}): these quickly converge to a neighbourhood of $\cA^-_{\beta,\sigma}$ for slow ramping ($\epsilon=0.001$) and various noise amplitudes $\sigma$. We plot $x$ against $\beta$ and show is the location of the hysteresis curve for $\epsilon=\sigma=0$, namely the unstable fixed points of (\ref{eq:unforced}) which satisfy $\beta=x-\alpha\arctan(x)$. Note that in the absence of noise $\sigma=0$ the system tracks the deterministic lower attractor to the tipping point and only transitions after that. For progressively larger amplitude noise and fixed $\epsilon$ we find progressively earlier transitions.

\begin{figure}
    \centering
    \includegraphics[width=15cm]{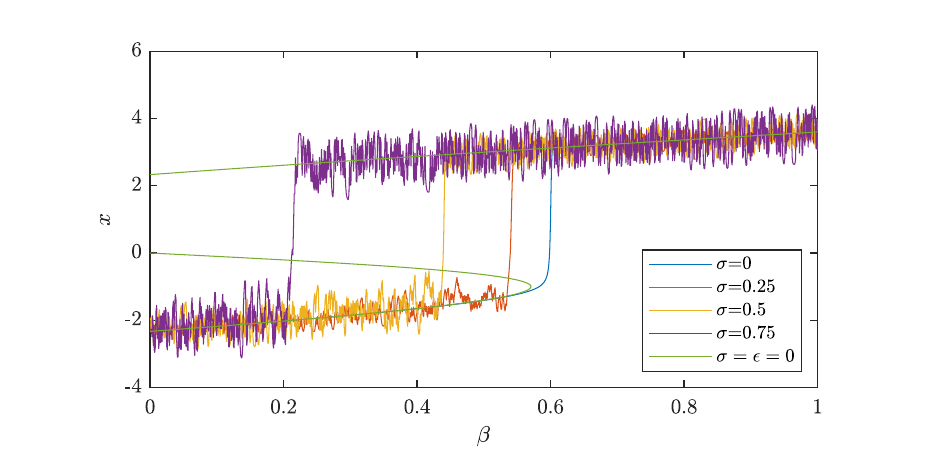} 
    
    \caption{Plots of typical trajectories of (\ref{eq:namap}) with $\beta_n$ ramped at a constant rate $\epsilon=0.001$ and a range of $\sigma$, for $\alpha=2$. The fixed points of (\ref{eq:unforced}) are also shown. Observe the delay of tipping to after the bifurcation in the smallest $\sigma$, but an advance tipping in the larger $\sigma$.}
    \label{fig:ramp}
\end{figure}

Figure~\ref{fig:upo} explores the dependence of the response of (\ref{eq:namap}) on varying rate of ramping $\epsilon$ for a fixed noise amplitude $\sigma=0.2$. Considering a variety of UPOs of $g$ with periods up to six we plot (in colours) the trajectories of the ramped system corresponding to the lower attractor as $\beta$ increases. We show a typical trajectory in bold black, and the hysteresis curve of the unforced fixed points $\beta(x)$ in non-bold black. The dashed black lines correspond to $\beta(x)\pm \sigma$.

There is a non-autonomous saddle-node bifurcation at the turning point of the leftmost black dotted line, where $\beta=\beta^*-\sigma$. This is the start of the chaotic tipping window and for typical $\theta$ the tipping can occur close to there by choosing arbitrarily small $\epsilon$. This corresponds also to the case of forcing by the extreme cases of the UPO that determine start and end of the chaotic tipping window. More precisely, the monotonicity \eqref{eq:monotonicity} and boundedness of the $x$ dynamics holds for \eqref{eq:namap} even in the case where $\beta_n$ varies\footnote{Or indeed $\sigma$!} with $n$. In consequence, for any given trajectory $\theta_n$ of $g$ we can find bounding trajectories $(x^+_n,\theta_n)$ and $(x^-_n,\theta_n)$ for $n\geq 1$ 
of \eqref{eq:namap} such that for any trajectory $(x_n,\theta_n)$ of \eqref{eq:namap} we have bounds as in \eqref{eq:attracting}.

The sequence Figure~\ref{fig:upo}(a)-(c) shows the effect of progressively faster ramping: all of the UPO forced trajectories tip over a range of values of $\beta$ around the unforced value $\beta^*$, and this range becomes biased to larger $\beta$ as $\epsilon$ increases. The typical trajectory (bold black) tips at some point that also becomes progressively later for $\epsilon$ larger. This means we can define a {\em dynamic tipping window}; the smallest interval that contains the set of all $\beta$ such that tipping for the non-autonomous system occurs at this $\beta$ with chaotic forcing typical for some measure $s\in\cS(m)$. Figure~\ref{fig:dtw} shows a numerical approximation of this dynamic tipping window for fixed $\sigma$ and varying $\epsilon$. Unlike the chaotic tipping window, the dynamic tipping window will depend on a chosen threshold to identify tipping. In our case we say there has been tipping from the lower branch at time $n$ if $x_n\geq 1$. However, the dynamic tipping window clearly limits to the chaotic tipping window for $\epsilon\rightarrow 0$.

\begin{figure}
    \centering

    (a)\includegraphics[align=t,width=12cm]{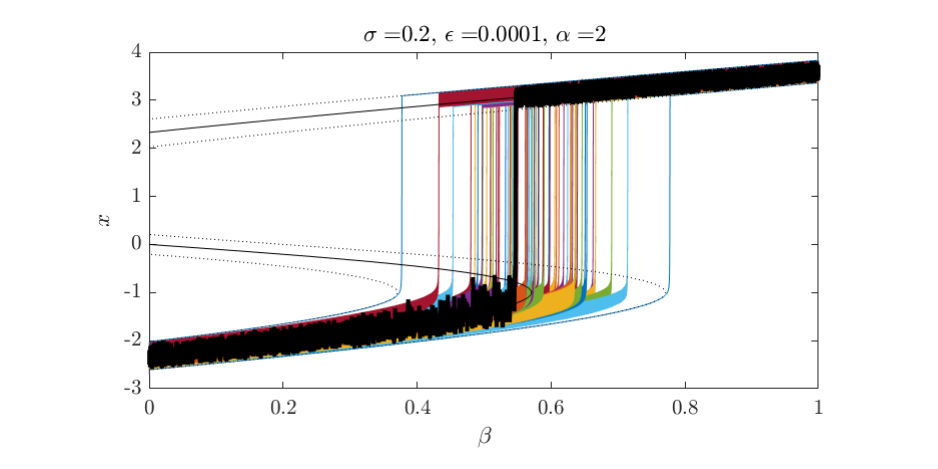} 

    (b)\includegraphics[align=t,width=12cm]{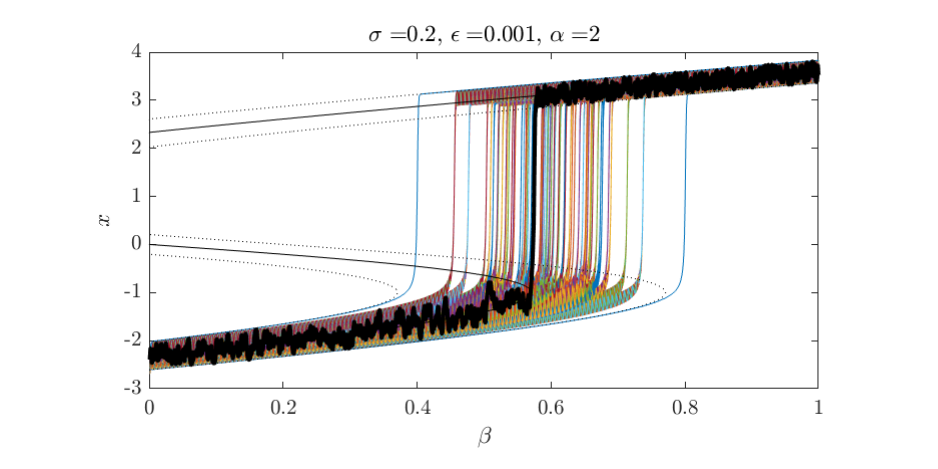}

    (c)\includegraphics[align=t,width=12cm]{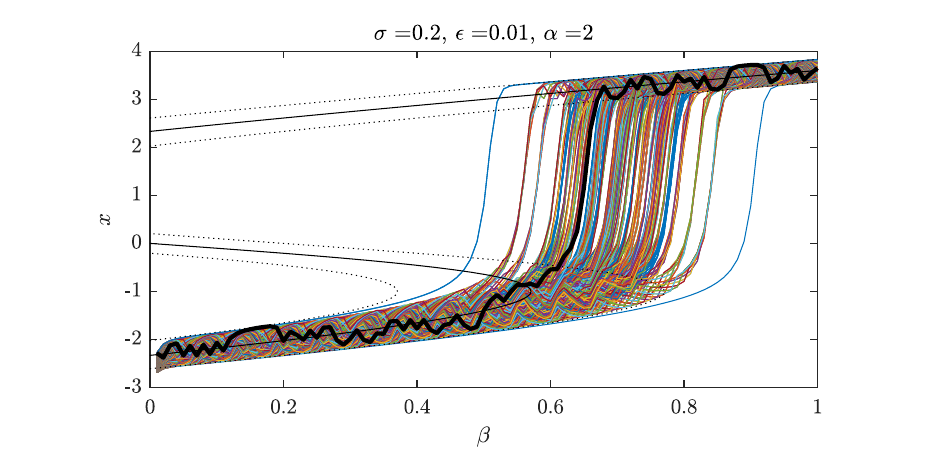} 

    \caption{Plots of trajectories of (\ref{eq:namap}) with $\beta_n$, $\sigma=0.2$ and $\alpha=2$, for a variety of UPOs of period up to period 6. The fixed points of (\ref{eq:unforced}) are also shown. A typical trajectory with chaotic forcing and these parameters is also shown (bold black line). We use $\beta_n=n\epsilon$ for (a) $\epsilon=0.0001$, (b) $\epsilon=0.001$ and (c) $\epsilon=0.01$.}
    \label{fig:upo}
\end{figure}

\begin{figure}
    \centering

    \includegraphics[align=t,width=13cm]{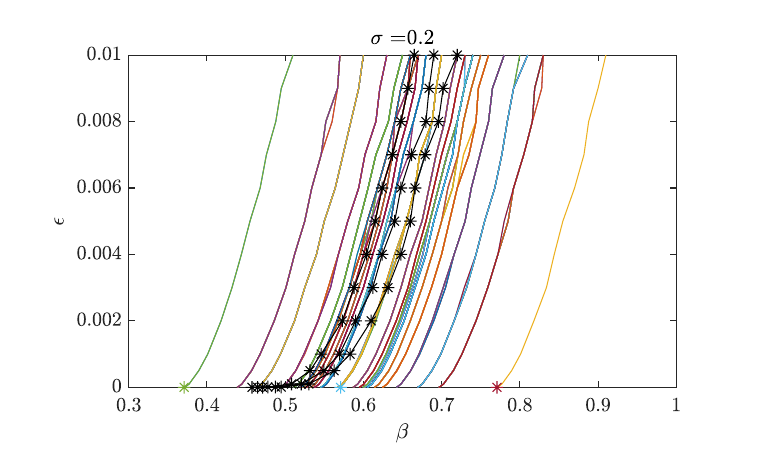}

    \caption{The dynamic tipping window for $\sigma=0.2$ and $\alpha=2$ for a range of values of $\epsilon\geq 10^{-5}$. The values of $\beta$ at tipping are shown for a number of UPOs of period up to period 5. The chaotic tipping window for $\epsilon=0$ is bounded by the values of $\beta$ at tipping for fixed points of (\ref{eq:unforced}) shown with coloured stars. The distribution of tipping  $\beta$ for an ensemble of $300$ values of $\theta$ from a uniform distribution is shown in terms of median and lower/upper quartiles (black stars). Note that the median passes the unforced tipping $\beta$ (blue star) at around $\epsilon=0.001$ (cf.\ Figure~\ref{fig:ll}).}
    \label{fig:dtw}
\end{figure}

In principle, typical trajectories can tip arbitrarily close to the non-autonomous saddle node, but because of super persistent chaotic transients that will exist close to a boundary crisis \cite{grebogi1985super,grebogi1987critical,lai2011transient} this will only be for $\epsilon$ that is exponentially small with respect to $\sigma$.

This is explored further in Figure~\ref{fig:cdfs} which shows numerically computed cdfs of the value of $\beta$ at tipping, measured as the value of $\beta$ where $x=1$ is first exceeded. Observe this cannot occur before the red line which indicates the start of a dynamic tipping window computed with forcing that is an extremal UPO. There is a competition between the rate of ramping $\epsilon$ and the amplitude of the noise $\sigma$: when $\epsilon$ is small then the median tipping location occurs before the unforced threshold, while when it exceed some rate the tipping location can be delayed. The relationship between $\sigma$, $\epsilon$ and the median time of tipping is shown in Figure~\ref{fig:ll}, together with the corresponding cdfs for the value of $\beta$ at tipping.

Analogous to the scaling for the ``strong noise" and ``weak noise" cases that are distinguished in \cite[Section 3.3]{berglund2006noise}, for $\epsilon\ll \sigma^2$ the median tipping will be before $\beta^*$ and so the chaotic forcing dominates, while for $\epsilon\gg\sigma^2$ the median tipping will be after $\beta^+$ and the bifurcation dominates. It is somewhat surprising that the dividing line of the scaling is so precise when considering the median tipping $\beta$ though this is almost certainly due to the special form of the map.

\begin{figure}
    \centering

    \includegraphics[align=t,width=14cm]{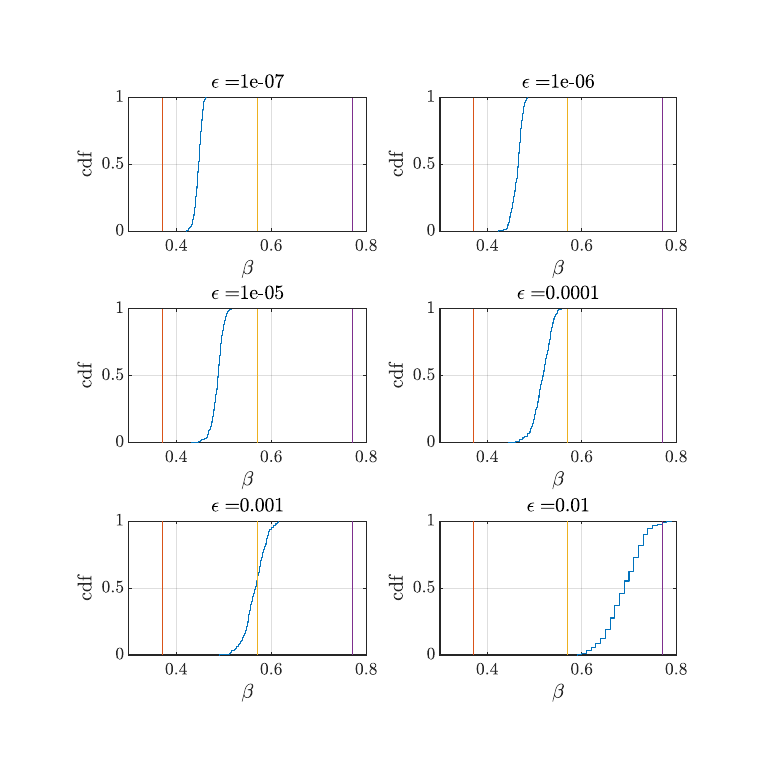}  

    \caption{Cumulative density functions (blue) of $\beta$ where tipping occurs for an ensemble of $N=300$ initial conditions, for $\sigma=0.2$ $\alpha=2$ and $\beta_n$ increasing by $\epsilon$ per timestep, as in Figure~\ref{fig:upo}. We say the tipping occurs when $x$ first exceeds $1$. The vertical lines show the locations of the chaotic tipping window (first tipping for UPO forcing) in red, and last tipping for UPO forcing in purple. The tipping for $\sigma=0$ is shown as vertical yellow line. Observe that for the smallest $\epsilon$ the median is close to the start of the chaotic tipping window.}
    \label{fig:cdfs}
\end{figure}

\begin{figure}
    \centering

    \includegraphics[align=t,width=9cm]{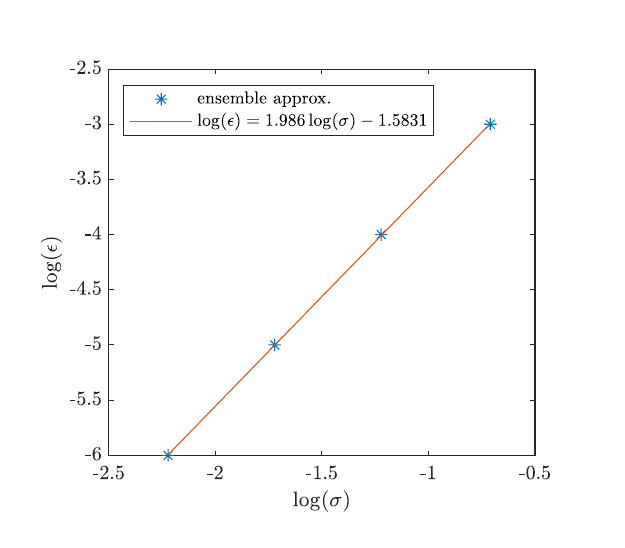}\includegraphics[align=t,width=8cm]{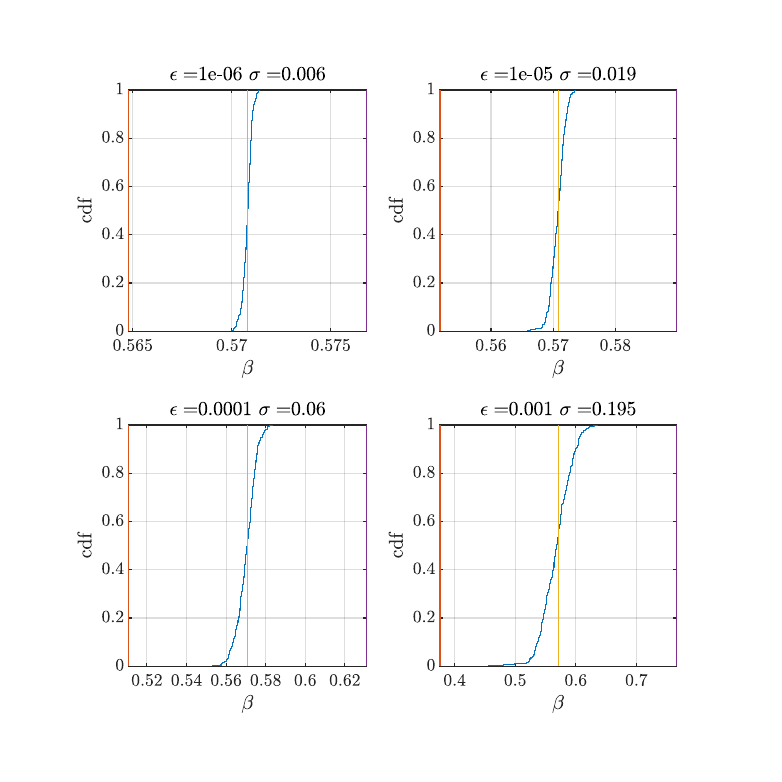}

    \caption{Left: relationship between $\sigma$ and $\epsilon$ such that $\beta_m$, the median of distribution of chaos-forced tipping, is equal to $\beta^*$. Note the close fit to $\epsilon\sim \sigma^2$ for the case where $\beta_m$ is an unbiased estimator of $\beta^*$. Right: the cdfs of four ensembles of $N=300$ initial conditions are shown within the dynamic tipping windows (as in Figure~\ref{fig:cdfs}). Note that the steps become discretized for large $\epsilon$ because only $\beta_n=n\epsilon$ are considered here.}
    \label{fig:ll}
\end{figure}

\section{Discussion}
\label{sec:discuss}

In this paper we have explored the effect of bounded {\em chaotic forcing} as a model for noise that we argue may in many applications be closer to the truth than the {\em stochastic forcing}, i.e. unbounded Wiener noise that can be obtained in a limit of fast chaos. 

By analysing an explicit example of a discrete time bistable map with chaotic forcing we are able to understand the interaction between multiple attractors and bounded noise forcing and the appearance of a chaotic tipping window. The start of this window can be seen as a boundary crisis \cite{lai2011transient} of an attractor of the skew product system or as a non-autonomous saddle-node bifurcation of the forced system. This becomes a dynamic tipping window if there is a slow variation of parameters in addition to chaotic forcing.

The problem of precisely which situations of forcing give rise to tipping and at what time is likely to be quite complicated in the presence of non-uniform parameter variation, especially when the timescale of change of $\beta$ is faster than the timescales associated with attraction to and mixing on the chaotic attractor. The chaos-forced transitions discussed can be related to bifurcations in set-valued dynamical systems such as investigated in \cite{lamb2015topological} as a model for bounded noise forced systems.

In this paper, although we initially discuss continuous time models with chaotic forcing, we do not consider any explicit examples for continuous time models with chaos. These have been considered elsewhere, for example \cite{ashwin2021physical,axelsen2023finite} consider a bistable (Stommel) system forced by a chaotic (Lorenz) system as a conceptual model for tipping of the Atlantic Meridional Overturning Circulation (AMOC) in response to fast chaotic variability, looking at cases where the system limits to an autonomous system in the past and future limits. 

Some of the results here can presumably be generalized to more general cases of non-autonomous parameter variation. For example, considering $\beta_n=\max(0,\min(n\epsilon,\beta^{\max}))$ can be investigated using examples as in Section~\ref{sec:ramped}. Such a case is analogous to a discrete ``parameter shift'' between $\beta=0$ and $\beta=\beta^{\max}>0$ at a rate $\epsilon$ \cite{Ashwinetal:2012,Ashwinetal:2017}. By picking $\theta_0$ and $x_0$ on one of the attractors for $\beta=0$ and then fixing $\beta^{\max}$ we expect to find different behaviours depending on where the chaotic tipping window lies in relation to $[0,\beta^{\max}]$. 

The tipping thresholds in general will involve considering extreme measures of the forcing map that may not be so simple as in this example: see \cite{jenkinson2019ergodic} for a recent review. For the example considered, one can include chaotic forcing of arbitrary amplitude by varying the parameter $\sigma$. Nonetheless, the large positive Lyapunov exponent $\ln(3)$ of the forcing chaos is always fast with respect to the dynamics of forced map unless $|\beta|$ is large. There are likely to be other effects present when the chaotic forcing becomes comparable or weaker than the attraction for the forced map, and it may be interesting to investigate such cases. We do not discuss here the problem of finding early warnings of an impending tipping point in this context. However, the observation that the fibre LE does not need to be zero at the crisis suggests that methods based purely on LEs are likely to miss prediction of an impending tipping point. See also for example \cite{kuehn2011mathematical,kuehn2018early}. Conversely, as suggested in \cite{nishikawa2014controlling}, it may be possible to use the detailed dynamics of the forced system to control the location of the tipping point or the behaviour after tipping. This will effectively require coupling from the forced system back to the forcing, and we do not consider this here.

\subsection*{Acknowledgements}

PA and JN received funding from the European Union’s Horizon 2020 research and innovation programme under grant agreement No. 820970 (TiPES). RR received funding from  the European Union’s Horizon 2020 Research and Innovation Programme under the Marie Skłodowska-Curie Grant Agreement No. 956170 (CriticalEarth).

\subsection*{Data availability}

The Matlab and Julia code to generate the figures in this paper is available from {\tt https://github.com/peterashwin/tipping-windows-2024}.

\bibliographystyle{plain}
\bibliography{nonautrefs}

\clearpage
\newpage

\appendix
\section{Proofs of dynamical properties}
\label{app:proofs}

The system~\eqref{eq:map} that we consider exhibits non-autonomous saddle-node bifurcations as studied in \cite{Anagnostopoulou:2012}. However, due to the special structure of our system, we can give more precise conclusions than in the general setup of \cite{Anagnostopoulou:2012}, using elementary proofs.

\subsection{Preliminaries}

By way of general notational convention, given an injective function $f$, we will write $f^{-1}$ to denote the inverse of $f$ defined on the range of $f$. For $x \in \bbR$ and $\delta>0$, $B_\delta(x):=(x-\delta,x+\delta)$.

Recall that $[0,2\pi)$ is endowed with the topology of a circle, and that $g \colon [0,2\pi) \to [0,2\pi)$ is the tripling map, $g(\theta)=3\theta ~(\bmod ~2\pi)$. We define the notations below where we fix $\alpha>1$.
\begin{align*}
    f_\beta \colon \overline{\bbR} &\to \bbR \\
    f_\beta(x) &= \alpha\atan(x) + \beta \qquad \forall \, \beta \in \bbR, \, x \in \overline{\bbR} \\ & \\
    f' \colon \bbR &\to \bbR \\
    f'(x) &= \frac{\alpha}{1+x^2} = f_\beta'(x) \qquad \textrm{(indep.\ of $\beta$)} \\ & \\
    f_{\beta,\sigma,\theta} &= f_{\beta+\sigma\cos(\theta)} \qquad \forall \, \beta \in \bbR, \, \sigma \geq 0, \, \theta \in [0,2\pi) \\ & \\
    f_{\beta,\sigma,\theta,n} &= f_{\beta,\sigma,g^{n-1}(\theta)} \circ \ldots \circ f_{\beta,\sigma,g(\theta)} \circ f_{\beta,\sigma,\theta} \qquad \forall \, n \in \mathbb{N}_0 \\ & \\
    G_{\beta,\sigma} \colon \bbR \times [0,2\pi) &\to \bbR \times [0,2\pi) \\
    G_{\beta,\sigma}(x,\theta) &= (f_{\beta,\sigma,\theta}(x),g(\theta)).
\end{align*}
So $G_{\beta,\sigma}$ is precisely the function mapping $(x_n,\theta_n)$ to $(x_{n+1},\theta_{n+1})$ in the difference equation~\eqref{eq:map}. Note that for all $n \in \mathbb{N}_0$ and $(x,\theta) \in \bbR \times [0,2\pi)$,
\[ G_{\beta,\sigma}^n(x,\theta) = (f_{\beta,\sigma,\theta,n}(x),g^n(\theta)). \]

Note that the function $f_0=\alpha\atan$ is an odd function, and that for each $\beta$, the function $f_\beta=f_0+\beta$ is strictly increasing on $\overline{\bbR}$, strictly convex on $(-\infty,0]$, and strictly concave on $[0,\infty)$, and that for all $y$ in the range of $f_\beta$,
\begin{equation} \label{eq:App_tan} f_\beta^{-1}(y) = \tan\!\left( \tfrac{y-\beta}{\alpha} \right). \end{equation}

Recall the $\beta$-parameterised families $a^\pm_\beta$ and $r_\beta$ described in Sec.~\ref{sec:Toy}:
\begin{itemize}
    \item for $\beta \in [-\beta^*,\infty)$, $f_\beta$ has a unique fixed point $a^+_\beta$ in $[x^*,\infty)$;
    \item for $\beta \in (-\infty,\beta^*]$, $f_\beta$ has a unique fixed point $a^-_\beta$ in $(-\infty,-x^*]$;
    \item for $\beta \in [-\beta^*,\beta^*]$, $f_\beta$ has a unique fixed point $r_\beta$ in $[-x^*,x^*]$;
\end{itemize}
where $x^*>0$ is the unique nonnegative root of $f'=1$ (so $-x^*$ is the unique nonpositive root of $f'=1$), and $\beta^*>0$ is the critical $\beta$-value for the saddle-node bifurcation at which $a^-_\beta$ and $r_\beta$ collide (so $-\beta^*$ is the critical $\beta$-value for the saddle-node bifurcation at which $a^+_\beta$ and $r_\beta$ collide). The fixed points of $f_\beta$ and their stability can be visualised from Fig.~\ref{fig:xmap}. Note that $x^*$ and $\beta^*$ depend on $\alpha>1$.

Now $x^*$ a fixed point of $f_{-\beta^*}$, and so $[x^*,\infty]$ is mapped into itself by $f_{-\beta^*}$; and so we can define a function $\varphi \colon [0,\infty] \to [0,\infty]$ by
\begin{align*}
    \varphi(x) &= f_{-\beta^*}(x^*+x) - x^* \\
    &= \alpha\atan(x^*+x) - \alpha\atan(x^*) \\
    &= f_\beta(x^*+x) - f_\beta(x^*) \qquad \textrm{for all } \beta.
\end{align*}
Note that $\varphi$ is increasing on $[0,\infty]$; and since $x^*$ is attracting from the right under $f_{-\beta^*}$, we have
\[ \varphi^n(\infty) = f_{-\beta^*}^n(\infty)-x^* \to 0 \quad \textrm{ as } n \to \infty. \]
Similarly, since $a_{-\beta^*}^-$ and $x^*$ are fixed points of $f_{-\beta^*}$, $[a_{-\beta^*}^-,x^*]$ is mapped bijectively onto itself by $f_{-\beta^*}$; and so we can define a function $\psi \colon [0,x^*-a_{-\beta^*}^-] \to [0,x^*-a_{-\beta^*}^-]$ by
\begin{align*}
    \psi(y) &= x^* - f_{-\beta^*}^{-1}(x^*-y) \\
    &= x^* - \tan\!\left( \atan(x^*) - \tfrac{y}{\alpha} \right) \\
    &= x^* - f_\beta^{-1}(f_\beta(x^*) - y) \qquad \textrm{for all } \beta.
\end{align*}
Note that $\psi$ is increasing on $[0,x^*-a_{-\beta^*}^-]$; and since, under $f_{-\beta^*}$, there is a heteroclinic connection from the left-sided-repelling point $x^*$ to the right-sided-attracting point $a_{-\beta^*}^-$, we have that for any $y \in [0,x^*-a_{-\beta^*}^-)$,
\[ \psi^n(y) = x^* - (f_{-\beta^*}^n)^{-1}(x^*-y) \to 0 \quad \textrm{ as } n \to \infty. \]

Let $\mathbb{P}$ be the normalised Lebesgue measure on $[0,2\pi)$. We have that $\mathbb{P}$ is a $g$-ergodic invariant probability measure. Hence in particular, $\mathbb{P}$ cannot be expressed non-trivially as a weighted average of $g$-invariant probability measures; consequently, since the ergodic decomposition theorem gives that every $G_{\beta,\sigma}$-invariant probability measure can be expressed as a weighted average of $G_{\beta,\sigma}$-ergodic invariant probability measures, we have the following:
\begin{lemma} \label{lemma:App_ergdec}
    For any Borel set $A \subset \bbR \times [0,2\pi)$,
    \begin{itemize}
        \item if there exists a $G_{\beta,\sigma}$-invariant probability measure with $\theta$-marginal $\mathbb{P}$ that assigns full measure to $A$, then there exists a $G_{\beta,\sigma}$-ergodic invariant probability measure with $\theta$-marginal $\mathbb{P}$ that assigns full measure to $A$;
        \item if there exists more than one $G_{\beta,\sigma}$-invariant probability measure with $\theta$-marginal $\mathbb{P}$ that assigns full measure to $A$, then there exists more than one $G_{\beta,\sigma}$-ergodic invariant probability measure with $\theta$-marginal $\mathbb{P}$ that assigns full measure to $A$.
    \end{itemize}
\end{lemma}

Let $\Theta$ denote the set of all two-sided orbits of $g$, that is
\[ \Theta = \{(\theta_i)_{i \in \mathbb{Z}} \in [0,2\pi)^\mathbb{Z} \, | \, \forall \, i \in \mathbb{Z}, \theta_{i+1} = g(\theta_i)\}. \]
Note that since $[0,2\pi)$ is compact under its topology as a circle and $g$ is continuous, $\Theta$ is compact. Let $\tilde{\mathbb{P}}$ be the unique probability measure on $\Theta$ with the property that for every $i \in \mathbb{Z}$, the $i$-th marginal of $\tilde{\mathbb{P}}$ is $\mathbb{P}$.

We define the $(\beta,\sigma)$-parameter regions $\rA$, $\rB^\pm$, $\rC^\pm$, $\rD$ and $S^\pm$ as in Sec.~\ref{sec:bulletpoints}.

\subsection{Construction of $\mu_{\beta,\sigma}^+$} \label{sec:A1}

For all $(\beta,\sigma,\theta) \in S^+ \times [0,2\pi)$, since
\[ \beta+\sigma\cos(\theta) \geq \beta-\sigma \geq -\beta^*, \]
we have
\[ f_{\beta,\sigma,\theta}(x^*) \geq f_{\beta-\sigma}(x^*) \geq f_{-\beta^*}(x^*) = x^* \]
and so
\[ f_{\beta,\sigma,\theta}([x^*,\infty]) \subset [x^*,\infty). \]
Therefore, for all $(\beta,\sigma) \in S^+$ and $\boldsymbol{\theta}=(\theta_i)_{i \in \mathbb{Z}} \in \Theta$, noting that
\[ f_{\beta,\sigma,\theta_{-(n+1)},n+1} = f_{\beta,\sigma,\theta_{-n},n} \circ  f_{\beta,\sigma,\theta_{-(n+1)}} \quad \forall \, n \geq 0, \]
we have (by induction) that the sequence of sets $f_{\beta,\sigma,\theta_{-n},n}([x^*,\infty])$ is nested, with diameter bounded above by $\varphi^n(\infty)$; and since $\varphi^n(\infty) \to 0$ as $n \to \infty$, the intersection of this sequence of sets is a singleton $\{a^+_{\beta,\sigma}(\boldsymbol{\theta})\}$.

For each $(\beta,\sigma) \in S^+$, we define the probability measure $\mu_{\beta,\sigma}^+$ on $\mathbb{R} \times [0,2\pi)$ to be the pushforward of $\tilde{\mathbb{P}}$ under the map $\Pi^+_{\beta,\sigma} \colon \boldsymbol{\theta} \mapsto (a^+_{\beta,\sigma}(\boldsymbol{\theta}),\theta_0)$, where $\theta_0$ denotes the $0$-coordinate of $\boldsymbol{\theta}$. Note that the $\theta$-marginal of $\mu_{\beta,\sigma}^+$ is the same as the $\theta_0$-marginal of $\tilde{\mathbb{P}}$, namely $\mathbb{P}$.

\subsection{Continuity of $a^+$ and $\Pi^+$} \label{sec:App_cont}

We now show that the map $(\beta,\sigma,\boldsymbol{\theta}) \mapsto a^+_{\beta,\sigma}(\boldsymbol{\theta})$ is continuous on $S^+ \times \Theta$; from this it immediately follows that the map $(\beta,\sigma,\boldsymbol{\theta}) \mapsto \Pi^+_{\beta,\sigma}(\boldsymbol{\theta})$ is also continuous on $S^+ \times \Theta$. Fix $(\beta,\sigma,\boldsymbol{\theta}) \in S^+ \times \Theta$, and fix $\varepsilon>0$. Let $N$ be such that
\[ f_{\beta,\sigma,\theta_{-N},N}([x^*,\infty]) \subset B_\varepsilon\big(a^+_{\beta,\sigma}(\boldsymbol{\theta})\big). \]
We can find a neighbourhood $U \subset \bbR \times [0,\infty) \times [0,2\pi)$ of $(\beta,\sigma,\theta_{-N})$ small enough that every $(\tilde{\beta},\tilde{\sigma},\tilde{\theta}) \in U$ has
\[ f_{\tilde{\beta},\tilde{\sigma},\tilde{\theta},N}([x^*,\infty]) \subset B_\varepsilon\big(a^+_{\beta,\sigma}(\boldsymbol{\theta})\big). \]
So then, every $(\tilde{\beta},\tilde{\sigma},\tilde{\boldsymbol{\theta}}) \in S^+ \times \Theta$ with $(\tilde{\beta},\tilde{\sigma},\tilde{\theta}_{-N}) \in U$ has
\[ a^+_{\tilde{\beta},\tilde{\sigma}}(\tilde{\boldsymbol{\theta}}) \in B_\varepsilon\big(a^+_{\beta,\sigma}(\boldsymbol{\theta})\big). \]

\subsection{Invariant-graph property of $a_{\beta,\sigma}^+$} \label{sec:App_inv}

For all $(\beta,\sigma) \in S^+$ and $(\theta_i)_i \in \Theta$,
\begin{align*}
    a^+_{\beta,\sigma}((\theta_{i+1})_i) &= \lim_{n \to \infty} f_{\beta,\sigma,\theta_{-(n-1)},n}(\infty) \\
    &= \lim_{n \to \infty} f_{\beta,\sigma,\theta_0}\Big( f_{\beta,\sigma,\theta_{-(n-1)},n-1}(\infty) \Big) \\
    &= f_{\beta,\sigma,\theta_0}\!\left( \lim_{n \to \infty} f_{\beta,\sigma,\theta_{-(n-1)},n-1}(\infty) \right) \qquad \textrm{since $f_{\beta,\sigma,\theta_0}$ is continuous} \\
    &= f_{\beta,\sigma,\theta_0}\Big( a^+_{\beta,\sigma}((\theta_i)_i) \Big).
\end{align*}

\subsection{Ergodicity and continuity properties of $\mu_{\beta,\sigma}^+$} \label{sec:App_ergcts}

Let $\tilde{g} \colon \Theta \to \Theta$ denote the shift map $(\theta_i)_i \mapsto (\theta_{i+1})_i$. Since $\mathbb{P}$ is $g$-ergodic, $\tilde{\mathbb{P}}$ is $\tilde{g}$-ergodic. By the invariant-graph property of $a_{\beta,\sigma}^+$ established in Sec.~\ref{sec:App_inv}, we have
\[ \Big( a^+_{\beta,\sigma}((\theta_{i+1})_i) \, , \, \theta_1 \Big) = G_{\beta,\sigma}\Big( a^+_{\beta,\sigma}((\theta_i)_i) \, , \, \theta_0 \Big) \]
for all $(\theta_i)_i \in \Theta$, i.e.
\[ \Pi^+_{\beta,\sigma} \circ \tilde{g} = G_{\beta,\sigma} \circ \Pi^+_{\beta,\sigma}. \]
In other words, $\Pi^+_{\beta,\sigma}$ is a morphism from $\tilde{g}$ to $G_{\beta,\sigma}$; consequently, since $\tilde{\mathbb{P}}$ is a $\tilde{g}$-ergodic invariant measure, it follows that $\mu_{\beta,\sigma}^+$ is a $G_{\beta,\sigma}$-ergodic invariant measure.

Due to the continuous dependence of $\Pi^+_{\beta,\sigma}(\boldsymbol{\theta})$ on $(\beta,\sigma)$ for each $\boldsymbol{\theta}$, the dominated convergence theorem gives that $\mu^+_{\beta,\sigma}$ depends continuously on $(\beta,\sigma)$ in the topology of weak convergence.

Since $\mu^+_{\beta,\sigma}$ is $G_{\beta,\sigma}$-invariant and $G_{\beta,\sigma}$ is continuous, we have that the support $\cA^+_{\beta,\sigma}$ of $\mu^+_{\beta,\sigma}$ is a  $G_{\beta,\sigma}$-invariant set. Moreover, due to the facts that $\Theta$ is compact and $\tilde{\mathbb{P}}$ has full support in $\Theta$ and $\Pi^+_{\beta,\sigma}$ is continuous, we have that $\cA^+_{\beta,\sigma}$ is precisely the image of $\Theta$ under $\Pi^+_{\beta,\sigma}$. Due to the continuous dependence of $\Pi^+_{\beta,\sigma}(\boldsymbol{\theta})$ on $(\beta,\sigma,\boldsymbol{\theta})$, it then follows that $\cA^+_{\beta,\sigma}$ has continuous dependence on $(\beta,\sigma)$ in Hausdorff distance.

\subsection{Bounds on the $x$-projection of $\cA^+_{\beta,\sigma}$} \label{sec:App_thetaspan}

Recall that we have shown that $\cA^+_{\beta,\sigma}$ is the image of $\Theta$ under $\Pi^+_{\beta,\sigma}$. We will show that
\begin{itemize}
    \item $\cA^+_{\beta,\sigma}$ is contained in $[a^+_{\beta-\sigma},a^+_{\beta+\sigma}] \times [0,2\pi)$, i.e.\ the range of $a^+_{\beta,\sigma}$ over $\Theta$ is contained in $[a^+_{\beta-\sigma},a^+_{\beta+\sigma}]$;
    \item if $\sigma>0$ then the intersection of $\cA_{\beta,\sigma}^+$ with $\cA_{\beta-\sigma}^+=\{a_{\beta-\sigma}^+\} \times [0,2\pi)$ is the point $(a^+_{\beta-\sigma},\pi)$ and the intersection of $\cA_{\beta,\sigma}^+$ with $\cA_{\beta+\sigma}^+=\{a_{\beta+\sigma}^+\} \times [0,2\pi)$ is the point $(a^+_{\beta+\sigma},0)$.
\end{itemize}
First note that for all $\beta \in [-\beta^*,\infty)$,
\[ a^+_\beta = \lim_{n \to \infty} f_\beta^n(\infty). \]
Now fix $(\beta,\sigma) \in S^+$, i.e.\ $\beta-\sigma \in [-\beta^*,\infty)$. So, since
\[ f_{\beta-\sigma} \leq f_{\beta,\sigma,\theta} \leq f_{\beta+\sigma} \]
for all $\theta$, we have that for every $\boldsymbol{\theta} \in \Theta$,
\begin{align*}
a^+_{\beta-\sigma} = \lim_{n \to \infty} f_{\beta-\sigma}^n(\infty) \leq \lim_{n \to \infty} f_{\beta,\sigma,\theta_{-n},n}(\infty) &= a^+_{\beta,\sigma}(\boldsymbol{\theta}) \\
a^+_{\beta+\sigma} = \lim_{n \to \infty} f_{\beta+\sigma}^n(\infty) \geq \lim_{n \to \infty} f_{\beta,\sigma,\theta_{-n},n}(\infty) &= a^+_{\beta,\sigma}(\boldsymbol{\theta}).
\end{align*}
Hence $\cA^+_{\beta,\sigma}$ is contained in $[a^+_{\beta-\sigma},a^+_{\beta+\sigma}] \times [0,2\pi)$.

Now note that $0$ and $\pi$ are fixed points of $g$. So the points $(0)_{i \in \mathbb{Z}}$ and $(\pi)_{i \in \mathbb{Z}}$ are members of $\Theta$, and we have
\begin{align*}
    a^+_{\beta,\sigma}((0)_{i \in \mathbb{Z}}) &= \lim_{n \to \infty} f_{\beta,\sigma,0,n}(\infty) = \lim_{n \to \infty} f_{\beta+\sigma}^n(\infty) = a^+_{\beta+\sigma} \\
    a^+_{\beta,\sigma}((\pi)_{i \in \mathbb{Z}}) &= \lim_{n \to \infty} f_{\beta,\sigma,\pi,n}(\infty) = \lim_{n \to \infty} f_{\beta-\sigma}^n(\infty) = a^+_{\beta-\sigma}.
\end{align*}
Hence the points $(a^+_{\beta+\sigma},0)$ and $(a^+_{\beta-\sigma},\pi)$ belong to $\cA^+_{\beta,\sigma}$. Now supposing that $\sigma>0$, we will show that these points are the only points in the intersection of $\cA^+_{\beta,\sigma}$ respectively with $\cA_{\beta+\sigma}^+$ and with $\cA_{\beta-\sigma}^+$. For this, it is sufficient to show that
\begin{itemize}
    \item[(i)] for every $\boldsymbol{\theta} \in \Theta$ with $\theta_0 \neq \pi$, $a^+_{\beta,\sigma}(\boldsymbol{\theta}) > a^+_{\beta-\sigma}$;
    \item[(ii)] for every $\boldsymbol{\theta} \in \Theta$ with $\theta_0 \neq 0$, $a^+_{\beta,\sigma}(\boldsymbol{\theta}) < a^+_{\beta+\sigma}$.
\end{itemize}
We show (i); (ii) is similar. Fix any $(\theta_i)_i \in \Theta$ with $\theta_0 \neq \pi$. Since $\pi$ is a fixed point of $g$, it follows that $\theta_{-1} \neq \pi$. So since $\sigma>0$ and $\theta_{-1} \neq \pi$, we have that $f_{\beta,\sigma,\theta_{-1}}$ is strictly greater than $f_{\beta-\sigma}$ everywhere. So, using the invariant-graph property of $a^+_{\beta,\sigma}$ applied at the point $(\theta_{i-1})_i$ we have
\begin{equation} \label{eq:neq0} a^+_{\beta,\sigma}((\theta_i)_i) = f_{\beta,\sigma,\theta_{-1}}\Big( a^+_{\beta,\sigma}((\theta_{i-1})_i) \Big) \geq f_{\beta,\sigma,\theta_{-1}}(a^+_{\beta-\sigma}) > f_{\beta-\sigma}(a^+_{\beta-\sigma}) = a^+_{\beta-\sigma}. \end{equation}

\subsection{Lyapunov exponent of $\mu^+_{\beta,\sigma}$} \label{sec:App_LE}

Note that $f'$ is continuous and always takes strictly positive values. So, since $\mu_{\beta,\sigma}^+$ is compactly supported, Birkhoff's ergodic theorem gives that the Lyapunov exponent $\lambda_{\beta,\sigma}(\mu_{\beta,\sigma}^+)$ is well-defined and given by
\[ \lambda_{\beta,\sigma}(\mu_{\beta,\sigma}^+) = \int_{\mathbb{R} \times [0,2\pi)} \ln f'(x) \, \mu_{\beta,\sigma}^+(d(x,\theta)) = \int_\Theta \ln f'(a_{\beta,\sigma}^+(\boldsymbol{\theta})) \, \tilde{\mathbb{P}}(d\boldsymbol{\theta}) . \]
Suppose $\sigma>0$. Then as in \eqref{eq:neq0}, every $\boldsymbol{\theta} \in \Theta$ outside the $\tilde{\mathbb{P}}$-null set $\{\theta_{-1} = \pi\}$ has $a_{\beta,\sigma}^+(\boldsymbol{\theta}) > a^+_{\beta-\sigma} \geq x^*$ and therefore
\[ f'(a_{\beta,\sigma}^+(\boldsymbol{\theta})) < 1; \]
and so it follows that
\[ \lambda_{\beta,\sigma}(\mu_{\beta,\sigma}^+) < 0. \]
Since $\ln f'$ is continuous and the support of $\mu^+_{\beta,\sigma}$ is uniformly bounded over bounded subsets of $S^+$, the continuous dependence of $\mu^+_{\beta,\sigma}$ on $(\beta,\sigma)$ in the topology of weak convergence implies that $\lambda_{\beta,\sigma}(\mu_{\beta,\sigma}^+)$ has continuous dependence on $(\beta,\sigma)$.

\subsection{Attractivity of $\cA^+_{\beta,\sigma}$} \label{sec:App_attrA+}

\begin{lemma}
    Fix $(\beta,\sigma) \in S^+$. The set $[x^*,\infty) \times [0,2\pi)$ is uniformly attracted to $\cA^+_{\beta,\sigma}$.
\end{lemma} \label{lemma:App_attr}

\begin{proof}
    For notational convenience, given $x,\tilde{x} \in \bbR$ and $\theta \in [0,2\pi)$, we write
    \[ \mathrm{dist}\big((x,\theta),(\tilde{x},\theta)\big) = |x-\tilde{x}|. \]
    To show the desired result, we will show that for every $(x,\theta) \in [x^*,\infty) \times [0,2\pi)$ there exists $\tilde{x}$ with $(\tilde{x},\theta) \in \cA^+_{\beta,\sigma}$ such that for all $n \geq 0$,
    \[ \mathrm{dist}\big(G_{\beta,\sigma}^n(x,\theta), G_{\beta,\sigma}^n(\tilde{x},\theta)\big) \leq \varphi^n(\infty). \]    
    Since $\cA^+_{\beta,\sigma}$ is $G_{\beta,\sigma}$-invariant and $\varphi^n(\infty) \to 0$ as $n \to \infty$, this gives the desired result.
    
    Fix $(x,\theta) \in [x^*,\infty) \times [0,2\pi)$. Pick a point $\boldsymbol{\theta} \in \Theta$ whose $0$-coordinate is equal to $\theta$, and set $\tilde{x}=a^+_{\beta,\sigma}(\boldsymbol{\theta})$; so $(\tilde{x},\theta) \in \cA^+_{\beta,\sigma}$. We have that both $x$ and $\tilde{x}$ belong to $[x^*,\infty)$; and so for all $n \geq 0$, since the diameter of $f_{\beta,\sigma,\theta,n}([x^*,\infty])$ is at most $\varphi^n(\infty)$, we have the desired result.
\end{proof}

Now if $\beta-\sigma$ is strictly greater than $-\beta^*$ then the minimal $x$-coordinate of $\cA^+_{\beta,\sigma}$, namely $a^+_{\beta-\sigma}$ (as shown in Sec.~\ref{sec:App_thetaspan}), is strictly greater than $x^*$; and so we immediately have the following corollary:

\begin{cor} \label{cor:App_localattr}
    If $\beta-\sigma$ is strictly greater than $-\beta^*$, then $\cA^+_{\beta,\sigma}$ is a local attractor.
\end{cor}

\subsection{Construction and properties of $\mu_{\beta,\sigma}^-$} \label{sec:App_mu-}

Let us first verify that the involution $\iota \colon (x,\theta) \mapsto (-x,\theta+\pi)$ serves as a conjugacy between $G_{\beta,\sigma}$ and $G_{-\beta,\sigma}$:
\begin{align*}
    \big(G_{\beta,\sigma} \circ \iota\big)(x,\theta) &= G_{\beta,\sigma}(-x,\theta+\pi) \\
    &= (\alpha \atan(-x) + \beta + \sigma\cos(\theta+\pi) \, , \, 3(\theta+\pi)) \\
    &= (-\alpha \atan(x) + \beta - \sigma\cos(\theta) \, , \, 3\theta+3\pi) \\
    &= \big(-(\alpha \atan(x) - \beta +\sigma\cos(\theta)) \, , \, 3\theta+\pi\big) \\
    &= \iota\big( \alpha \atan(x) - \beta +\sigma\cos(\theta) \, , \, 3\theta \big) \\
    &= \big(\iota \circ G_{-\beta,\sigma}\big)(x,\theta).
\end{align*}
Since
\[ S^- = \{(\beta,\sigma) : (-\beta,\sigma) \in S^+\}, \]
for every $(\beta,\sigma) \in S^-$ we can define $\mu_{\beta,\sigma}^-$ to be the pushforward of $\mu_{-\beta,\sigma}^+$ under $\iota$. Since $\iota$ is a conjugacy between $G_{\beta,\sigma}$ and $G_{-\beta,\sigma}$, and $\mu_{-\beta,\sigma}^+$ is a $G_{-\beta,\sigma}$-ergodic measure, we have that $\mu_{\beta,\sigma}^-$ is a $G_{\beta,\sigma}$-ergodic measure; and for all the other properties that we have established for $\mu_{\beta,\sigma}^+$, the analogous properties for $\mu_{\beta,\sigma}^-$ hold.

\subsection{Construction of $r_{\beta,\sigma}$} \label{sec:App_rconstr}

Rather than seeking to define $\cC_{\beta,\sigma}$ directly in terms of the basins of attraction of $\cA^\pm_{\beta,\sigma}$, we will first present a construction of the function $r_{\beta,\sigma} \colon [0,2\pi) \to \bbR$, and we will define $\cC_{\beta,\sigma}$ to be the graph of $r_{\beta,\sigma}$; with these definitions, we will prove all the required properties, including that when both $\cA^-_{\beta,\sigma}$ and $\cA^+_{\beta,\sigma}$ exist, $\cC_{\beta,\sigma}$ is the basin boundary of both $\cA^-_{\beta,\sigma}$ and $\cA^+_{\beta,\sigma}$.

Recall that the parameter region~$\rA$ is defined to be $S^+ \cap S^-$. For all $(\beta,\sigma,\theta) \in \rA \times [0,2\pi)$, since $\beta-\sigma \geq -\beta^*$ we have as in Sec.~\ref{sec:A1} that
\[ f_{\beta,\sigma,\theta}(x^*) \geq x^*; \]
and similarly, since $\beta+\sigma \leq \beta^*$, we have
\[ f_{\beta,\sigma,\theta}(-x^*) \leq f_{\beta+\sigma}(-x^*) \leq f_{\beta^*}(-x^*)=-x^*. \]
Combining these two facts, we have that
\[ [-x^*,x^*] \subset f_{\beta,\sigma,\theta}([-x^*,x^*]) \]
and so
\[ f_{\beta,\sigma,\theta}^{-1}([-x^*,x^*]) \subset [-x^*,x^*]. \]
Therefore, noting that for all $n \geq 0$,
\[ f_{\beta,\sigma,\theta,n+1}^{-1} = f_{\beta,\sigma,\theta,n}^{-1} \circ f_{\beta,\sigma,g^n(\theta)}^{-1} \]
on the range of $f_{\beta,\sigma,\theta,n+1}$, we have (by induction) that the sequence of sets $f_{\beta,\sigma,\theta,n}^{-1}([-x^*,x^*])$ is nested, with diameter bounded above by $\psi^n(2x^*)$. Since $a_{-\beta^*}^-$ is strictly less than $-x^*$, we have that $2x^*$ is strictly less than $x^*-a_{-\beta^*}^-$ and so $\psi^n(2x^*)$ tends to $0$ as $n \to \infty$. Hence, the intersection of the nested sequence of sets $f_{\beta,\sigma,\theta,n}^{-1}([-x^*,x^*])$ is a singleton $\{r_{\beta,\sigma}(\theta)\}$.

\subsection{Properties of $r_{\beta,\sigma}$ analogous to Secs.~\ref{sec:A1}--\ref{sec:App_LE}} \label{sec:App_prop_r}

For all $(\beta,\sigma) \in \rA$, let $\Psi_{\beta,\sigma} \colon [0,2\pi) \to \bbR \times [0,2\pi)$ be the map $\Psi_{\beta,\sigma}(\theta)=(r_{\beta,\sigma}(\theta),\theta)$, and let $\nu_{\beta,\sigma}$ be the pushforward of $\mathbb{P}$ under $\Psi_{\beta,\sigma}$.

Similarly to in Sec.~\ref{sec:App_cont}, we show that the map $(\beta,\sigma,\theta) \mapsto r_{\beta,\sigma}(\theta)$ is continuous on $\rA \times [0,2\pi)$. Fix $(\beta,\sigma,\theta) \in \rA \times [0,2\pi)$ and $\varepsilon>0$. Let $N$ be such that
\[ f_{\beta,\sigma,\theta,N}^{-1}([-x^*,x^*]) \subset B_\varepsilon\big(r_{\beta,\sigma}(\theta)\big), \]
i.e.
\[ [-x^*,x^*] \subset \Big( f_{\beta,\sigma,\theta,N}(r_{\beta,\sigma}(\theta)-\varepsilon) \, , \, f_{\beta,\sigma,\theta,N}(r_{\beta,\sigma}(\theta)+\varepsilon)\Big). \]
So we can find a neighbourhood $U  \subset \bbR \times [0,\infty) \times [0,2\pi)$ of $(\beta,\sigma,\theta)$ small enough that every $(\tilde{\beta},\tilde{\sigma},\tilde{\theta}) \in U$ has
\[ f_{\tilde{\beta},\tilde{\sigma},\tilde{\theta},N}^{-1}([-x^*,x^*]) \subset B_\varepsilon\big(r_{\beta,\sigma}(\theta)\big). \]
So then, every $(\tilde{\beta},\tilde{\sigma},\tilde{\theta})$ in the intersection of $\rA \times [0,2\pi)$ with $U$ has
\[ r_{\tilde{\beta},\tilde{\sigma}}(\tilde{\theta}) \in B_\varepsilon\big(r_{\beta,\sigma}(\theta)\big). \]
Thus we have shown that the map $(\beta,\sigma,\theta) \mapsto r_{\beta,\sigma}(\theta)$ is continuous; and it follows that the map $(\beta,\sigma,\theta) \mapsto \Psi_{\beta,\sigma}(\theta)$ is likewise continuous.

Similarly to in Sec.~\ref{sec:App_inv}, we show that $r_{\beta,\sigma}$ has the invariant-graph property: for all $(\beta,\sigma,\theta) \in \rA \times [0,2\pi)$,
\begin{align*}
    r_{\beta,\sigma}(g(\theta)) &= \lim_{n \to \infty} f_{\beta,\sigma,g(\theta),n}^{-1}(x^*) \\
    &= \lim_{n \to \infty} f_{\beta,\sigma,\theta}\Big( f_{\beta,\sigma,\theta,n+1}^{-1}(x^*) \Big) \\
    &= f_{\beta,\sigma,\theta}\!\left( \lim_{n \to \infty} f_{\beta,\sigma,\theta,n+1}^{-1}(x^*) \right) \qquad \textrm{since $f_{\beta,\sigma,\theta}$ is continuous} \\
    &= f_{\beta,\sigma,\theta}\big( r_{\beta,\sigma}(\theta) \big).
\end{align*}
Similarly to in Sec.~\ref{sec:App_ergcts}, it follows that $\Psi_{\beta,\sigma}$ is a morphism from $g$ to $G_{\beta,\sigma}$, and hence, since $\mathbb{P}$ is a $g$-ergodic invariant measure, we have that $\nu_{\beta,\sigma}$ is a $G_{\beta,\sigma}$-ergodic invariant measure. As in Sec.~\ref{sec:App_ergcts}, we have that: $\nu_{\beta,\sigma}$ depends continuously on $(\beta,\sigma)$ in the topology of weak convergence; the graph $\cC_{\beta,\sigma}:=\{(r_{\beta,\sigma}(\theta),\theta) : \theta \in [0,2\pi)\}$ of $r_{\beta,\sigma}$ is precisely the support of $\nu_{\beta,\sigma}$; and $\cC_{\beta,\sigma}$ has continuous dependence on $(\beta,\sigma)$ in Hausdorff distance.

Similarly to in Sec.~\ref{sec:App_thetaspan}, we show that the range of $r_{\beta,\sigma}$ is contained in $[r_{\beta+\sigma},r_{\beta-\sigma}]$. Since $(\beta,\sigma) \in \rA$, we have that both $\beta-\sigma$ and $\beta+\sigma$ lie in the interval $[-\beta^*,\beta^*]$, and so $r_{\beta+\sigma}$ and $r_{\beta-\sigma}$ both exist. Now note that for all $\beta \in [-\beta^*,\beta^*]$, since $x^* \in [r_\beta,a_\beta^+]$, we have
\[ r_\beta = \lim_{n \to \infty} (f_\beta^n)^{-1}(x^*). \]
Since (by Eq.~\eqref{eq:App_tan})
\[ f_{\beta+\sigma}^{-1}(y) \leq f_{\beta,\sigma,\theta}^{-1}(y) \leq f_{\beta-\sigma}^{-1}(y) \]
for all $\theta \in [0,2\pi)$ and $y \in [-x^*,x^*]$, we have that for all $\theta \in [0,2\pi)$,
\begin{align*}
    r_{\beta-\sigma} &= \lim_{n \to \infty} (f_{\beta-\sigma}^n)^{-1}(x^*) \geq \lim_{n \to \infty} f_{\beta,\sigma,\theta,n}^{-1}(x^*) = r_{\beta,\sigma}(\theta) \\
    r_{\beta+\sigma} &= \lim_{n \to \infty} (f_{\beta+\sigma}^n)^{-1}(x^*) \leq \lim_{n \to \infty} f_{\beta,\sigma,\theta,n}^{-1}(x^*) = r_{\beta,\sigma}(\theta).
\end{align*}
Hence, the range of $r_{\beta,\sigma}$ is contained in $[r_{\beta+\sigma},r_{\beta-\sigma}]$.

Next, similarly to in Sec.~\ref{sec:App_thetaspan}, we show that if $\sigma>0$ then $r_{\beta,\sigma}$ attains the value $r_{\beta-\sigma}$ uniquely at $\theta=\pi$, and attains the value $r_{\beta+\sigma}$ uniquely at $\theta=0$. Since $0$ and $\pi$ are fixed points of $g$, we have
\begin{align*}
    r_{\beta,\sigma}(0) &= \lim_{n \to \infty} f_{\beta,\sigma,0,n}^{-1}(x^*) = \lim_{n \to \infty} (f_{\beta+\sigma}^n)^{-1}(x^*) = r_{\beta+\sigma} \\
    r_{\beta,\sigma}(\pi) &= \lim_{n \to \infty} f_{\beta,\sigma,\pi,n}^{-1}(x^*) = \lim_{n \to \infty} (f_{\beta-\sigma}^n)^{-1}(x^*) = r_{\beta-\sigma}.
\end{align*}
Now assuming $\sigma>0$, every $\theta$ distinct from $\pi$ has $f_{\beta,\sigma,\theta}^{-1}$ being strictly less than $f_{\beta-\sigma}^{-1}$ on $[-x^*,x^*]$, and therefore
\[ r_{\beta,\sigma}(\theta) = f_{\beta,\sigma,\theta}^{-1}\Big( r_{\beta,\sigma}(g(\theta)) \Big) \leq f_{\beta,\sigma,\theta}^{-1}(r_{\beta-\sigma}) < f_{\beta-\sigma}^{-1}(r_{\beta-\sigma}) = r_{\beta-\sigma}. \]
A similar argument gives that every $\theta$ distinct from $0$ has $r_{\beta,\sigma}(\theta)>r_{\beta+\sigma}$.

Similarly to in Sec.~\ref{sec:App_LE}, we have
\[ \lambda_{\beta,\sigma}(\nu_{\beta,\sigma}) = \frac{1}{2\pi} \int_0^{2\pi} \ln f'(r_{\beta,\sigma}(\theta)) \, d\theta. \]
Every $\theta$ outside of $\{0,\pi\}$ has $r_{\beta,\sigma}(\theta)$ lying in the open interval $(r_{\beta+\sigma},r_{\beta-\sigma})$, and hence in the open interval $(-x^*,x^*)$, and so
\[ f'(r_{\beta,\sigma}(\theta)) > 1; \]
and so it follows that
\[ \lambda_{\beta,\sigma}(\nu_{\beta,\sigma}) > 0. \]
As in Sec.~\ref{sec:App_LE}, $\lambda_{\beta,\sigma}(\nu_{\beta,\sigma})$ has continuous dependence on $(\beta,\sigma)$.

\subsection{Intersections of $\cC_{\beta,\sigma}$ with $\cA^\pm_{\beta,\sigma}$} \label{sec:App_inters}

Drawing together the results in Sec.~\ref{sec:App_thetaspan} and the analogous results in Sec.~\ref{sec:App_prop_r}:

For $(\beta,\sigma)$ in the interior of $\rA$ relative to $\bbR \times [0,\infty)$, we have that
\begin{itemize}
    \item $\beta-\sigma>-\beta^*$ and so $r_{\beta-\sigma}<x^*<a^+_{\beta-\sigma}$;
    \item $\beta+\sigma<\beta^*$ and so $a^-_{\beta+\sigma}<-x^*<r_{\beta+\sigma}$;
\end{itemize}
and so
\begin{itemize}
    \item the $x$-projection of $\cA^+_{\beta,\sigma}$ (namely, $[a^+_{\beta-\sigma},a^+_{\beta+\sigma}]$) lies strictly to the right of $x^*$;
    \item the $x$-projection of $\cA^-_{\beta,\sigma}$ (namely, $[a^-_{\beta-\sigma},a^-_{\beta+\sigma}]$) lies strictly to the left of $-x^*$;
    \item the $x$-projection of $\cC_{\beta,\sigma}$ (namely, $[r_{\beta+\sigma},r_{\beta-\sigma}]$) lies strictly between $-x^*$ and $x^*$;
\end{itemize}
and hence $\cC_{\beta,\sigma}$ has empty intersection with $\cA^+_{\beta,\sigma}$ and with $\cA^-_{\beta,\sigma}$.

Now note that
\[ L^\pm = \{ (\beta,\sigma) \in \rA \, : \, \beta \pm \sigma = \pm \beta^* \}. \]
So then:
\begin{itemize}
    \item For $(\beta,\sigma) \in L^-$, we have $r_{\beta-\sigma}=a^+_{\beta-\sigma}=x^*$, and so both $\cA^+_{\beta,\sigma}$ and $\cC_{\beta,\sigma}$ include the point $(x^*,\pi)$. If, furthermore, $\sigma>0$, then this is the only point of intersection of $\cA^+_{\beta,\sigma}$ or $\cC_{\beta,\sigma}$ with $\{x^*\} \times [0,2\pi)$, and hence is the only point of intersection of $\cA^+_{\beta,\sigma}$ with $\cC_{\beta,\sigma}$.
    \item For $(\beta,\sigma) \in L^+$, we have $r_{\beta+\sigma}=a^-_{\beta+\sigma}=-x^*$, and so both $\cA^-_{\beta,\sigma}$ and $\cC_{\beta,\sigma}$ include the point $(-x^*,0)$. If, furthermore, $\sigma>0$, then this is the only point of intersection of $\cA^+_{\beta,\sigma}$ or $\cC_{\beta,\sigma}$ with $\{-x^*\} \times [0,2\pi)$, and hence is the only point of intersection of $\cA^-_{\beta,\sigma}$ with $\cC_{\beta,\sigma}$.
\end{itemize}

\subsection{Basin of attraction of $\cA^\pm_{\beta,\sigma}$ for $(\beta,\sigma) \in \rA$} \label{sec:App_BB}

\begin{prop} \label{prop:App_bb}
    For $(\beta,\sigma) \in \rA$,
    \begin{itemize}
        \item[(i)] every $(x,\theta)$ with $x>r_{\beta,\sigma}(\theta)$ is in the basin of $\cA^+_{\beta,\sigma}$;
        \item[(ii)] every $(x,\theta)$ with $x<r_{\beta,\sigma}(\theta)$ is in the basin of $\cA^-_{\beta,\sigma}$;
        \item[(iii)] hence $\cA^+_{\beta,\sigma}$ and $\cA^-_{\beta,\sigma}$ share the same basin boundary, namely $\cC_{\beta,\sigma}$.
    \end{itemize}
\end{prop}

\begin{proof}
Point~(iii) follows from points~(i) and (ii) due to the continuity of $r_{\beta,\sigma}$. We will prove (i); the proof of (ii) is similar.

Fix $(\beta,\sigma) \in \rA$, and fix $(x,\theta)$ with $x>r_{\beta,\sigma}(\theta)$. Take $N$ such that $f_{\beta,\sigma,\theta,N}^{-1}(x^*) \leq x$. Then $x^* \leq f_{\beta,\sigma,\theta,N}(x)$, i.e.\ $G_{\beta,\sigma}^N(x,\theta) \in [x^*,\infty) \times [0,2\pi)$, and so Lemma~\ref{lemma:App_attr} gives that $(x,\theta)$ is in the basin of $\cA^+_{\beta,\sigma}$.
\end{proof}

\subsection{Basin of attraction of $\cA^+_{\beta,\sigma}$ for $(\beta,\sigma) \in S^+ \setminus \rA$} \label{sec:App_Bfull}

By construction, the region
\[ S^+ \setminus \rA \, = \, S^+ \setminus S^- \, = \, \{(\beta,\sigma) \in \bbR \times [0,\infty) \, : \, \beta-\sigma \geq -\beta^* \ \textrm{ and } \ \beta+\sigma>\beta^* \} \]
partitions into region~$\rB^+$ and region~$\rC^+$. Now
\[ (\beta,\sigma) \not\in \rB^+ \quad \Longleftrightarrow \quad \beta-\sigma \leq \beta^*, \]
and so $\rC^+$ is the set of all members of $S^+ \setminus \rA$ that fulfil this, i.e.
\[ \rC^+ = \{(\beta,\sigma) \in \bbR \times [0,\infty) \, : \, -\beta^* \leq \beta-\sigma \leq \beta^* < \beta+\sigma\}. \]

\begin{prop} \label{prop:App_full}
~

    \begin{itemize}
        \item[(I)] For each $(\beta,\sigma) \in \rC^+$, the basin of $\cA^+_{\beta,\sigma}$ includes both the set $(r_{\beta-\sigma},\infty) \times [0,2\pi)$ and a set of the form $\bbR \times E$ for some $\mathbb{P}$-full measure set $E \subset [0,2\pi)$.
        \item[(II)] For each $(\beta,\sigma) \in \rB^+$, $\cA^+_{\beta,\sigma}$ is a global attractor.
    \end{itemize}
\end{prop}

\begin{proof}
(I) Fix $(\beta,\sigma) \in \rC^+$. Since $\beta-\sigma \in [-\beta^*,\beta^*]$, we have in particular that $r_{\beta-\sigma}$ exists. Now for any $(x,\theta)$ with $x>r_{\beta-\sigma}$:
\begin{itemize}
    \item If $\beta-\sigma=-\beta^*$, then $r_{\beta-\sigma}=x^*$, so $x>x^*$, and so Lemma~\ref{lemma:App_attr} gives that $(x,\theta)$ is in the basin of $\cA^+_{\beta,\sigma}$.
    \item If $\beta-\sigma>-\beta^*$ then
    \[ x^* < a^+_{\beta-\sigma} = \lim_{n \to \infty} f_{\beta-\sigma}^n(x) \leq \liminf_{n \to \infty} f_{\beta,\sigma,\theta,n}(-\infty), \]
    and so there are $n$ for which $f_{\beta,\sigma,\theta,n}(x)>x^*$; so Lemma~\ref{lemma:App_attr} gives that $(x,\theta)$ is in the basin of $\cA^+_{\beta,\sigma}$.
\end{itemize}
Thus we have shown that the basin of $\cA^+_{\beta,\sigma}$ includes $(r_{\beta-\sigma},\infty) \times [0,2\pi)$.

Now since $\beta+\sigma>\beta^*$ (and therefore, also, $\beta+\sigma>-\beta^*$), we have
\[ x^* < a^+_{\beta+\sigma} = \lim_{n \to \infty} f_{\beta+\sigma}^n(-\infty) = \lim_{n \to \infty} f_{\beta,\sigma,0,n}(-\infty). \]
Therefore, we can find a positive integer $N$ and a neighbourhood $U$ of $0$ (in the space of angles) such that for all $\theta \in U$, $f_{\beta,\sigma,\theta,N}(-\infty) \geq x^*$. So $G^N(\bbR \times U) \subset [x^*,\infty) \times [0,2\pi)$. Therefore, by Lemma~\ref{lemma:App_attr}, any initial condition whose trajectory under $G$ enters the set $\bbR \times U$ belongs to the basin of $\cA^+_{\beta,\sigma}$. In other words, letting $E:=\bigcup_{n=0}^\infty g^{-n}(U)$, we have that $\bbR \times E$ is contained in the basin of $\cA^+_{\beta,\sigma}$. And since $\mathbb{P}$ is $g$-ergodic and $\mathbb{P}(U)>0$, we have $\mathbb{P}(E)=1$.

(II) Fix $(\beta,\sigma) \in \rB^+$. Since $\beta-\sigma>\beta^*$, we have
\[ \lim_{n \to \infty} f_{\beta-\sigma}^n(-\infty) = a_{\beta-\sigma}>x^* \]
and so we can find a positive integer $N$ such that for all $\theta \in [0,2\pi)$, $f_{\beta,\sigma,\theta,N}(-\infty) \geq f_{\beta-\sigma}^N(-\infty) \geq x^*$. So $G^N(\bbR \times [0,2\pi)) \subset [x^*,\infty) \times [0,2\pi)$, and hence Lemma~\ref{lemma:App_attr} gives that $\cA^+_{\beta,\sigma}$ is a global attractor.
\end{proof}

The analogous results for the basin of $\cA^-_{\beta,\sigma}$ with $(\beta,\sigma) \in S^- \setminus \rA$ also hold.

\subsection{Basin of attraction of $\mu^\pm_{\beta,\sigma}$} \label{sec:App_basinmu}

Fix $(\beta,\sigma) \in S^+$.

\begin{lemma} \label{lemma:App_contr}
    For every $(x,(\theta_i)_i) \in \bbR \times \Theta$ with $(x,\theta_0)$ in the basin of $\cA^+_{\beta,\sigma}$, we have that
    \[ \left| f_{\beta,\sigma,\theta_0,n}(x) - a^+_{\beta,\sigma}((\theta_{i+n})_i) \right| \to 0 \quad \textrm{as } n \to \infty. \]
\end{lemma}

\begin{proof}
Fix $(x,(\theta_i)_i) \in \bbR \times \Theta$ with $(x,\theta_0)$ in the basin of $\cA^+_{\beta,\sigma}$.

First suppose there exists $N \geq 0$ such that $f_{\beta,\sigma,\theta_0,N}(x) \geq x^*$; then since the range of $a^+_{\beta,\sigma}$ is also contained in $[x^*,\infty)$, it follows that for all $n \geq N$,
\begin{align*}
    \left| f_{\beta,\sigma,\theta_0,n}(x) - a^+_{\beta,\sigma}((\theta_{i+n})_i) \right| &= \left| f_{\beta,\sigma,\theta_N,n-N}\Big(f_{\beta,\sigma,\theta_0,N}(x)\Big) - f_{\beta,\sigma,\theta_N,n-N}\Big(a^+_{\beta,\sigma}((\theta_{i+N})_i)\Big) \right| \\
    &\leq \mathrm{diam}(f_{\beta,\sigma,\theta_N,n-N}([x^*,\infty])) \\
    &\leq \varphi^{n-N}(\infty).
\end{align*}
Hence, since $\varphi^{n-N}(\infty) \to 0$ as $n \to \infty$, we have the desired conclusion.

Now suppose instead that $f_{\beta,\sigma,\theta_0,n}(x) < x^*$ for all $n \geq 0$. As in Sec.~\ref{sec:App_inters}, if it were the case that $\beta-\sigma>-\beta^*$, then the $x$-projection of $\cA^+_{\beta,\sigma}$ would lie strictly to the right of $x^*$. So since $(x,\theta_0)$ in the basin of $\cA^+_{\beta,\sigma}$, we must have that $\beta-\sigma=-\beta^*$ (i.e.\ $a^+_{\beta-\sigma}=x^*$) and $\lim_{n \to \infty} f_{\beta,\sigma,\theta_0,n}(x) = x^*$. If $\sigma=0$ then on the whole of $\Theta$ we have
\[ a^+_{\beta,\sigma} = a^+_{-\beta^*,0} = x^* \]
and so the fact that $\lim_{n \to \infty} f_{\beta,\sigma,\theta_0,n}(x) = x^*$ is precisely the required result. So now assume that $\sigma>0$. So the only point in $\cA^+_{\beta,\sigma}$ whose $x$-coordinate is $x^*$ is $(x^*,\pi)$. So since $(x,\theta_0)$ in the basin of $\cA^+_{\beta,\sigma}$, we have that $\theta_n \to \pi$ as $n \to \infty$. Hence, due to the continuity of $a^+_{\beta,\sigma}$,
\[ \lim_{n \to \infty} a^+_{\beta,\sigma}((\theta_{i+n})_i) = a^+_{\beta,\sigma}((\pi)_i) = a^+_{\beta-\sigma} = x^*. \]
Combining this with the fact that $\lim_{n \to \infty} f_{\beta,\sigma,\theta_0,n}(x) = x^*$ gives the required result.
\end{proof}

\begin{prop} \label{prop:App_phys}
    There is a $\mathbb{P}$-full measure set $R \subset [0,2\pi)$ such that every initial condition $(x,\theta)$ in the basin of $\cA^+_{\beta,\sigma}$ with $\theta \in R$ is in the basin of $\mu^+_{\beta,\sigma}$.
\end{prop}

\begin{proof}
Define
\begin{align*}
    \tilde{G}_{\beta,\sigma} \colon \bbR \times \Theta &\to \bbR \times \Theta \\
    \tilde{G}_{\beta,\sigma}(x,(\theta_i)_i) &= (f_{\beta,\sigma,\theta_0}(x),(\theta_{i+1})_i)
\end{align*}
and
\begin{align*}
    \tilde{\Pi}^+_{\beta,\sigma} \colon \Theta &\to \bbR \times \Theta \\
    \tilde{\Pi}^+_{\beta,\sigma}(\boldsymbol{\theta}) &= (a^+_{\beta,\sigma}(\boldsymbol{\theta}),\boldsymbol{\theta}).
\end{align*}
Let $\tilde{\mu}^+_{\beta,\sigma}$ be the pushforward of $\tilde{\mathbb{P}}$ under $\tilde{\Pi}^+_{\beta,\sigma}$. By the same argument as in Sec.~\ref{sec:App_ergcts}, we have that $\tilde{\Pi}^+_{\beta,\sigma}$ is a morphism from $\tilde{g}$ to $\tilde{G}_{\beta,\sigma}$ and therefore $\tilde{\mu}^+_{\beta,\sigma}$ is a $\tilde{G}_{\beta,\sigma}$-ergodic invariant measure. So $\tilde{\mu}^+_{\beta,\sigma}$-almost every point is in the basin of $\tilde{\mu}^+_{\beta,\sigma}$ under $\tilde{G}_{\beta,\sigma}$. Thus, in other words, for $\tilde{\mathbb{P}}$-almost every $\boldsymbol{\theta}$, $\tilde{\Pi}^+_{\beta,\sigma}(\boldsymbol{\theta})$ is in the basin of $\tilde{\mu}^+_{\beta,\sigma}$ under $\tilde{G}_{\beta,\sigma}$. So let $\Theta_1 \subset \Theta$ be a $\tilde{\mathbb{P}}$-full measure set such that for every $\boldsymbol{\theta} \in \Theta_1$, $\tilde{\Pi}^+_{\beta,\sigma}(\boldsymbol{\theta})$ is in the basin of $\tilde{\mu}^+_{\beta,\sigma}$ under $\tilde{G}_{\beta,\sigma}$. Let $R \subset [0,2\pi)$ be the image of $\Theta_1$ under the projection $(\theta_i)_i \mapsto \theta_0$. So $R$ is a $\mathbb{P}$-full measure subset of $[0,2\pi)$.

Fix $(x,\theta)$ in the basin of $\cA^+_{\beta,\sigma}$ with $\theta \in R$, and fix a bounded uniformly continuous function $h \colon \bbR \times [0,2\pi) \to \bbR$. Take $(\theta_i)_i \in \Theta_1$ with $\theta_0=\theta$. Write $\tilde{h}(x',(\theta_i')_i)=h(x',\theta_0')$ for all $(x',(\theta_i')_i) \in \bbR \times \Theta$. Then for all $N \geq 1$,
\begin{align*}
    d_N :=& \ \int_{\bbR \times [0,2\pi)} h \, d\mu^+_{\beta,\sigma} - \frac{1}{N} \sum_{n=0}^{N-1} h(G_{\beta,\sigma}^n(x,\theta)) \\
    =& \ \int_{\bbR \times \Theta} \tilde{h} \, d\tilde{\mu}^+_{\beta,\sigma} - \frac{1}{N} \sum_{n=0}^{N-1} h(f_{\beta,\sigma,\theta_0,n}(x),\theta_n) \\
    =& \ \Bigg( \int_{\bbR \times \Theta} \tilde{h} \, d\tilde{\mu}^+_{\beta,\sigma} - \frac{1}{N} \sum_{n=0}^{N-1} h(a^+_{\beta,\sigma}((\theta_{i+n})_i),\theta_n) \Bigg) \\ &\hspace{20mm} + \frac{1}{N} \sum_{n=0}^{N-1}\bigg( h(f_{\beta,\sigma,\theta_0,n}(x),\theta_n) - h(a^+_{\beta,\sigma}((\theta_{i+n})_i),\theta_n) \bigg) \\
    =& \ \Bigg( \underbrace{\int_{\bbR \times \Theta} \tilde{h} \, d\tilde{\mu}^+_{\beta,\sigma} - \frac{1}{N} \sum_{n=0}^{N-1} \tilde{h}(\tilde{G}_{\beta,\sigma}^n(\tilde{\Pi}^+_{\beta,\sigma}((\theta_i)_i)))}_{A_N} \Bigg) \\ &\hspace{20mm} + \frac{1}{N} \! \sum_{n=0}^{N-1}\bigg( \! \underbrace{h(f_{\beta,\sigma,\theta_0,n}(x),\theta_n) - h(a^+_{\beta,\sigma}((\theta_{i+n})_i),\theta_n)}_{B_n} \bigg).
\end{align*}
Since $\tilde{\Pi}^+_{\beta,\sigma}((\theta_i)_i))$ is in the basin of $\tilde{\mu}^+_{\beta,\sigma}$ under $\tilde{G}_{\beta,\sigma}$ and $\tilde{h}$ is bounded and continuous, we have that $A_N \to 0$ as $N \to \infty$. Since $h$ is uniformly continuous, Lemma~\ref{lemma:App_contr} gives that $B_n \to 0$ as $n \to \infty$. Thus $d_N \to 0$ as $N \to \infty$, i.e.\ $(x,\theta)$ is in the basin of $\mu^+_{\beta,\sigma}$.
\end{proof}

By Lemma~\ref{lemma:App_attr}, the basin of $\cA^+_{\beta,\sigma}$ includes $[x^*,\infty) \times [0,2\pi)$; hence it follows from Proposition~\ref{prop:App_phys} that $\mu^+_{\beta,\sigma}$ is a physical measure and is the only ergodic invariant measure with support contained in $[x^*,\infty) \times [0,2\pi)$ whose $\theta$-marginal is $\mathbb{P}$. Moreover, if $(\beta,\sigma) \in S^+ \setminus \rA$, then it follows from Proposition~\ref{prop:App_phys} and Proposition~\ref{prop:App_full}(I) that $\mu^+_{\beta,\sigma}$ is the only ergodic invariant measure whose $\theta$-marginal is $\mathbb{P}$.

(Each of the above two statements about uniqueness of an ergodic invariant measure then translates into a statement about uniqueness of an invariant measure by Lemma~\ref{lemma:App_ergdec}.)

The analogous results for $\mu^-_{\beta,\sigma}$ holds.

\subsection{The $x$-projection of invariant-measure supports in $\rD$} \label{sec:App_Dinv}

First observe that $\rD$ is contained in the complement of $\rB^- \cup \rB^+$.

Given any $(\beta,\sigma)$, a compact set $S \subset \bbR \times [0,2\pi)$ is said to be \emph{positively invariant} under $G_{\beta,\sigma}$ if $G_{\beta,\sigma}(S)$ is a subset of $S$.

\begin{prop} \label{prop:App_big_ga}
    For each $(\beta,\sigma)$ outside of $\rB^- \cup \rB^+$, the set $[a^-_{\beta-\sigma},a^+_{\beta+\sigma}] \times [0,2\pi)$ is positively invariant and for every initial condition $(x,\theta) \in \bbR \times [0,2\pi)$, the distance of $G_{\beta,\sigma}^n(x,\theta)$ from $[a^-_{\beta-\sigma},a^+_{\beta+\sigma}] \times [0,2\pi)$ tends to $0$ as $n \to \infty$.
\end{prop}

\begin{proof}
Since $(\beta,\sigma) \not\in \rB^-$, we have that $\beta+\sigma \geq -\beta^*$, and so $a^+_{\beta+\sigma}$ exists and for all $\theta \in [0,2\pi)$,
\begin{equation} \label{eq:-bpm+}
\limsup_{n \to \infty} f_{\beta,\sigma,\theta,n}(\infty) \leq \lim_{n \to \infty} f_{\beta+\sigma}^n(\infty) = a^+_{\beta+\sigma}.
\end{equation}
Similarly, since $(\beta,\sigma) \not\in \rB^+$, we have that $\beta-\sigma \leq \beta^*$, and so $a^-_{\beta-\sigma}$ exists and for all $\theta \in [0,2\pi)$,
\begin{equation} \label{eq:-bpm-}
\liminf_{n \to \infty} f_{\beta,\sigma,\theta,n}(-\infty) \geq \lim_{n \to \infty} f_{\beta-\sigma}^n(-\infty) = a^-_{\beta-\sigma}.
\end{equation}
Due to the monotonicity of the maps $f_{\beta,\sigma,\theta,n}$, combining Eqs.~\eqref{eq:-bpm+} and \eqref{eq:-bpm-} gives the desired result.
\end{proof}

As a consequence of Proposition~\ref{prop:App_big_ga}, we have that for all $(\beta,\sigma)$ outside of $\rB^- \cup \rB^+$, the support of every $G_{\beta,\sigma}$-invariant probability measure is contained in $[a^-_{\beta-\sigma},a^+_{\beta+\sigma}] \times [0,2\pi)$. Moreover, since $[a^-_{\beta-\sigma},a^+_{\beta+\sigma}] \times [0,2\pi)$ is positively invariant and $[a^-_{\beta-\sigma},a^+_{\beta+\sigma}]$ is compact, the random Krylov-Bogolyubov theorem~\cite[Corollary~6.13]{Crauel_RPM} together with Lemma~\ref{lemma:App_ergdec} gives that there exists at least one $G_{\beta,\sigma}$-ergodic invariant measure with $\theta$-marginal $\mathbb{P}$.

\begin{prop} \label{prop:App_endpoints}
    For $(\beta,\sigma) \in \rD$, the support of every $G_{\beta,\sigma}$-invariant measure $\mu$ with $\theta$-marginal $\mathbb{P}$ includes the points $(a^-_{\beta-\sigma},\pi)$ and $(a^+_{\beta+\sigma},0)$.
\end{prop}

We start with a lemma.

\begin{lemma} \label{lemma:App_globcontr}
~
    \begin{itemize}
        \item[(I)] For any $(\beta,\sigma)$ outside $S^+$, for every $\theta_0 \in [0,2\pi)$ with $\pi \in \overline{\{g^i(\theta_0)\}_{i \geq 0}}$ and every neighbourhood $U \subset \bbR \times [0,2\pi)$ of $(a^-_{\beta-\sigma},\pi)$, there exists $n \geq 0$ such that $G_{\beta,\sigma}^n(\bbR \times \{\theta_0\})$ is contained in $U$.
        \item[(II)] For any $(\beta,\sigma)$ outside $S^-$, for every $\theta_0 \in [0,2\pi)$ with $0 \in \overline{\{g^i(\theta_0)\}_{i \geq 0}}$ and every neighbourhood $U \subset \bbR \times [0,2\pi)$ of $(a^+_{\beta+\sigma},0)$, there exists $n \geq 0$ such that $G_{\beta,\sigma}^n(\bbR \times \{\theta_0\})$ is contained in $U$.
    \end{itemize}
\end{lemma}

\begin{proof}
We prove (I); the proof of (II) is similar. Fix $(\beta,\sigma)$ outside $S^+$. So $\beta-\sigma<-\beta^*$, and so the singleton $\{a^-_{\beta-\sigma}\}$ is a global attractor of $f_{\beta-\sigma}$. Fix $\theta_0 \in [0,2\pi)$ with $\pi \in \overline{\{g^i(\theta_0)\}_{i \geq 0}}$. Without loss of generality, take $U$ to be the open square of side-length $2\varepsilon$ centred on $(a^-_{\beta-\sigma},\pi)$ for some $\varepsilon>0$.

Since $\pi$ is a fixed point of $g$, we can find $m \geq 0$ such that
\[ f_{\beta,\sigma,\pi,m}(\overline{\bbR}) = f_{\beta-\sigma}^m(\overline{\bbR}) \subset \big(a^-_{\beta-\sigma}-\varepsilon,a^-_{\beta-\sigma}+\varepsilon\big). \]
So, by the compactness of $\overline{\bbR}$ and the continuity of the map $(x,\theta) \mapsto f_{\beta,\sigma,\theta,m}(x)$ from $\overline{\bbR} \times [0,2\pi)$ to $\bbR$, we can then find $\delta>0$ such that for all $\theta \in (\pi-\delta,\pi+\delta)$,
\[ f_{\beta,\sigma,\theta,m}(\overline{\bbR}) \subset \big(a^-_{\beta-\sigma}-\varepsilon,a^-_{\beta-\sigma}+\varepsilon\big). \]
Using again the fact that $\pi$ is a fixed point of $g$, since $g$ is continuous we can find $\delta_1 \in (0,\delta)$ such that for every $\theta \in (\pi-\delta_1,\pi+\delta_1)$, $g^m(\theta) \in (\pi-\varepsilon,\pi+\varepsilon)$. Since $\pi \in \overline{\{g^i(\theta_0)\}_{i \geq 0}}$, we can find $n_0 \geq 0$ such that $g^{n_0}(\theta_0) \in (\pi-\delta_1,\pi+\delta_1)$. So now, letting $n:=n_0+m$, we have
\[ g^n(\theta_0) = g^m(g^{n_0}(\theta_0)) \in (\pi-\varepsilon,\pi+\varepsilon) \]
and
\[ f_{\beta,\sigma,\theta_0,n}(\bbR) \subset f_{\beta,\sigma,g^{n_0}(\theta_0),m}(\bbR) \subset \big(a^-_{\beta-\sigma}-\varepsilon,a^-_{\beta-\sigma}+\varepsilon\big), \]
and therefore
\[ G_{\beta,\sigma}^n(\bbR \times \{\theta_0\}) \subset U. \qedhere \]
\end{proof}

\begin{proof}[Proof of Proposition~\ref{prop:App_endpoints}]
Fix $(\beta,\sigma) \in \rD$, and fix a $G_{\beta,\sigma}$-invariant measure $\mu$ with $\theta$-marginal $\mathbb{P}$. Since $\mathbb{P}$ is a $g$-ergodic invariant probability measure, the trajectory of $\mathbb{P}$-almost every point $\theta_0$ under $g$ is dense in $[0,2\pi)$. So taking $(x_0,\theta_0) \in \mathrm{supp}\,\mu$ with $\theta_0$ having a dense trajectory, Lemma~\ref{lemma:App_globcontr}(I) gives that $(a^-_{\beta-\sigma},\pi)$ is in the closure of the $G_{\beta,\sigma}$-trajectory of $(x_0,\theta_0)$, and Lemma~\ref{lemma:App_globcontr}(II) gives that $(a^+_{\beta+\sigma},0)$ is in the closure of the $G_{\beta,\sigma}$-trajectory of $(x_0,\theta_0)$.
\end{proof}

\subsection{Globally contractive dynamics} \label{sec:App_GC}

We will need the following general fact:

\begin{lemma} \label{lemma:App_convprob}
    Let $(X_n)_{n \geq 1}$ be a uniformly bounded sequence of nonnegative-valued random variables that converges in probability to $0$. Then $\liminf_{N \to \infty} \frac{1}{N} \sum_{n=1}^N X_n$ is almost surely equal to $0$.
\end{lemma}

\begin{proof}
Since $(X_n)_n$ is uniformly bounded, so is $\big(\frac{1}{N} \sum_{n=1}^N X_n\big)_N$. Now for a uniformly bounded sequence of random variables, convergence in probability is equivalent to convergence in $L^1$; so
\[ 0 = \lim_{n \to \infty} \mathbb{E}[X_n] = \lim_{N \to \infty} \frac{1}{N} \sum_{n=1}^N \mathbb{E}[X_n] = \lim_{N \to \infty} \mathbb{E}\!\left[ \frac{1}{N} \sum_{n=1}^N X_n \right], \]
and therefore $\frac{1}{N} \sum_{n=1}^N X_n$ converges in probability to $0$ as $N \to \infty$. Since convergence in probability implies the existence of a subsequence on which we have almost sure convergence, the result follows.
\end{proof}

Now for any $(\beta,\sigma) \in \bbR \times [0,\infty)$, $n \geq 1$ and $\theta \in [0,2\pi)$, let
\[ X_{\beta,\sigma,n}(\theta) = f_{\beta,\sigma,\theta,n}(\infty) - f_{\beta,\sigma,\theta,n}(-\infty). \]
Similarly to in Sec.~\ref{sec:A1}, note that for each $(\theta_i)_{i \in \mathbb{Z}} \in \Theta$, the sequence of sets $f_{\beta,\sigma,\theta_{-n},n}(\overline{\bbR})$ is nested. So let
\[ \Theta_{\beta,\sigma} = \left\{ (\theta_i)_i \in \Theta \, : \, \bigcap_{n=0}^\infty f_{\beta,\sigma,\theta_{-n},n}(\overline{\bbR}) \textrm{ is a singleton} \right\}, \]
and for each $(\theta_i)_i \in \Theta_{\beta,\sigma}$, let $a_{\beta,\sigma}((\theta_i)_i) \in \bbR$ be such that
\[ \bigcap_{n=0}^\infty f_{\beta,\sigma,\theta_{-n},n}(\overline{\bbR}) = \{a_{\beta,\sigma}((\theta_i)_i)\}. \]
Now let us define $\tilde{\rD}$ to be the set of all $(\beta,\sigma) \in \bbR \times [0,\infty)$ for which $\tilde{\mathbb{P}}(\Theta_{\beta,\sigma})=1$; and similarly to in Sec.~\ref{sec:A1}, for each $(\beta,\sigma) \in \tilde{\rD}$, let $\mu_{\beta,\sigma}$ be the pushforward of $\tilde{\mathbb{P}}$ under the map $\Pi_{\beta,\sigma} \colon \boldsymbol{\theta} \mapsto (a_{\beta,\sigma}(\boldsymbol{\theta}),\theta_0)$, where $\theta_0$ denotes the $0$-coordinate of $\boldsymbol{\theta}$. We will now see that this matches the definitions of $\tilde{\rD}$ and $\mu_{\beta,\sigma}$ given in Sec.~\ref{sec:bulletpoints}.

\begin{prop} \label{prop:App_tildeD}
    For $(\beta,\sigma) \in \bbR \times [0,\infty)$, the following statements are equivalent:
    \begin{itemize}
        \item[(i)] $(\beta,\sigma) \in \tilde{\rD}$;
        \item[(ii)] over the probability space $([0,2\pi),\mathbb{P})$, $X_{\beta,\sigma,n}$ converges in probability to $0$ as $n \to \infty$;
        \item[(iii)] there is only one $G_{\beta,\sigma}$-invariant measure with $\theta$-marginal $\mathbb{P}$.
    \end{itemize}
    If these statements hold, $\mu_{\beta,\sigma}$ is the unique $G_{\beta,\sigma}$-invariant measure with $\theta$-marginal $\mathbb{P}$.
\end{prop}

\begin{proof}
For all $(\theta_i)_{i \in \mathbb{Z}} \in \Theta$ and $n \geq 1$, let
\[ \tilde{X}_{\beta,\sigma,n}((\theta_i)_i) = f_{\beta,\sigma,\theta_{-n},n}(\infty) - f_{\beta,\sigma,\theta_{-n},n}(\infty). \]
Define
\begin{align*}
    \dot{a}^-_{\beta,\sigma}((\theta_i)_i) &= \lim_{n \to \infty} f_{\beta,\sigma,\theta_{-n},n}(-\infty) \\
    \dot{a}^+_{\beta,\sigma}((\theta_i)_i) &= \lim_{n \to \infty} f_{\beta,\sigma,\theta_{-n},n}(\infty) \\
    \tilde{X}_{\beta,\sigma} &= \dot{a}^+_{\beta,\sigma} - \dot{a}^-_{\beta,\sigma} = \lim_{n \to \infty} \tilde{X}_{\beta,\sigma,n}
\end{align*}
So $\bigcap_{n=0}^\infty f_{\beta,\sigma,\theta_{-n},n}(\overline{\bbR})$ is precisely $[\dot{a}^-_{\beta,\sigma}((\theta_i)_i),\dot{a}^+_{\beta,\sigma}((\theta_i)_i)]$.

We first show the equivalence of (i) and (ii). For all $n \geq 1$, since $\mathbb{P}$ is precisely the $\theta_{-n}$-marginal of $\tilde{\mathbb{P}}$, the law of $\tilde{X}_{\beta,\sigma,n}$ over $(\Theta,\tilde{\mathbb{P}})$ is the same as the law of $X_{\beta,\sigma,n}$ over $([0,2\pi),\mathbb{P})$. Hence, 
as $n \to \infty$, $X_{\beta,\sigma,n}$ converges in law to the law of $\tilde{X}_{\beta,\sigma}$. But $\tilde{X}_{\beta,\sigma} \overset{\tilde{\mathbb{P}}\textrm{-a.s.}}{=} 0$ if and only if $\dot{a}^-_{\beta,\sigma} \overset{\tilde{\mathbb{P}}\textrm{-a.s.}}{=} \dot{a}^+_{\beta,\sigma}$, i.e.\ if and only if $(\beta,\sigma) \in \tilde{\rD}$.

We next show that (iii)$\Rightarrow$(i). Let $\dot{\mu}^\pm_{\beta,\sigma}$ be the pushforward of $\tilde{\mathbb{P}}$ under the map $\dot{\Pi}^\pm_{\beta,\sigma} \colon \boldsymbol{\theta} \mapsto (\dot{a}^\pm_{\beta,\sigma}(\boldsymbol{\theta}),\theta_0)$. Note that the $x$-marginal of $\dot{\mu}^\pm_{\beta,\sigma}$ is precisely the law of $\dot{a}^\pm_{\beta,\sigma}$ over $(\Theta,\tilde{\mathbb{P}})$. Now just as in Sec.~\ref{sec:App_inv}, $\dot{a}^\pm_{\beta,\sigma}$ fulfils the invariant-graph property; and therefore, as in Secs.~\ref{sec:A1} and \ref{sec:App_ergcts}, $\dot{\mu}^-_{\beta,\sigma}$ and $\dot{\mu}^+_{\beta,\sigma}$ are $G_{\beta,\sigma}$-invariant measures with $\theta$-marginal $\mathbb{P}$. Therefore, assuming (iii), we have $\dot{\mu}^-_{\beta,\sigma}=\dot{\mu}^+_{\beta,\sigma}$ and so the law of $\dot{a}^-_{\beta,\sigma}$ is the same as the law of $\dot{a}^+_{\beta,\sigma}$. But since $\dot{a}^-_{\beta,\sigma} \leq \dot{a}^+_{\beta,\sigma}$ throughout $\Theta$, it then follows that $\dot{a}^-_{\beta,\sigma} \overset{\tilde{\mathbb{P}}\textrm{-a.s.}}{=} \dot{a}^+_{\beta,\sigma}$.

To show (i)$\Rightarrow$(ii), suppose for a contradiction that (i) holds and (ii) does not. On the basis of Lemma~\ref{lemma:App_ergdec}, let $\mu_{\beta,\sigma}^1$ and $\mu_{\beta,\sigma}^2$ be distinct $G_{\beta,\sigma}$-ergodic invariant measures with $\theta$-marginal $\mathbb{P}$, and let $E_1,E_2 \subset [0,2\pi)$ be the $\theta$-projection of the basin of $\mu_{\beta,\sigma}^1,\mu_{\beta,\sigma}^2$ respectively. Now $E_1$ and $E_2$ are $\mathbb{P}$-full measure sets; and since $(X_{\beta,\sigma,n})_{n \geq 1}$ is uniformly bounded by $\alpha\pi$, Lemma~\ref{lemma:App_convprob} gives that $\liminf_{N \to \infty} \frac{1}{N} \sum_{n=1}^N X_{\beta,\sigma,n} \overset{\mathbb{P}\textrm{-a.s.}}{=} 0$. So fix a point $\theta \in E_1 \cap E_2$ such that $\liminf_{N \to \infty} \frac{1}{N} \sum_{n=1}^N X_{\beta,\sigma,n}(\theta)=0$. Fix $x_1,x_2$ such that $(\theta,x_1)$ and $(\theta,x_2)$ are in the basin of $\mu_{\beta,\sigma}^1$ and $\mu_{\beta,\sigma}^2$ respectively. Let $h \colon \bbR \times [0,2\pi)$ be a $1$-Lipschitz function such that $\int_{\bbR \times [0,2\pi)} h \, d\mu_{\beta,\sigma}^1 \neq \int_{\bbR \times [0,2\pi)} h \, d\mu_{\beta,\sigma}^2$. We have
\begin{align*}
    0 &< \left| \int_{\bbR \times [0,2\pi)} h \, d\mu_{\beta,\sigma}^1 - \int_{\bbR \times [0,2\pi)} h \, d\mu_{\beta,\sigma}^2 \right| \\
    &= \lim_{N \to \infty} \left| \frac{1}{N} \sum_{n=1}^N \bigg( h(G_{\beta,\sigma}^n(\theta,x_1)) - h(G_{\beta,\sigma}^n(\theta,x_2)) \bigg) \right| \\
    &\leq \liminf_{N \to \infty} \frac{1}{N} \sum_{n=1}^N |f_{\beta,\sigma,\theta,n}(x_1)-f_{\beta,\sigma,\theta_0,n}(x_2)| \\
    &\leq \liminf_{N \to \infty} \frac{1}{N} \sum_{n=1}^N X_{\beta,\sigma,n}(\theta) \\
    &=0,
\end{align*}
giving a contradiction.

If the statements (i)--(iii) hold, then as in the proof of the direction (iii)$\Rightarrow$(i), $\mu_{\beta,\sigma}$ is a $G_{\beta,\sigma}$-invariant measure with $\theta$-marginal $\mathbb{P}$.
\end{proof}

Now let $\tilde{\rD}_{\lambda<0}$ be as defined in Sec.~\ref{sec:bulletpoints}, namely, $\tilde{\rD}_{\lambda<0} \subset \tilde{\rD}$ is the set of all $(\beta,\sigma) \in \tilde{\rD}$ for which $\lambda_{\beta,\sigma}(\mu_{\beta,\sigma})<0$.

\begin{prop} \label{prop:App_expconv}
    For all $(\beta,\sigma) \in \tilde{\rD}_{\lambda<0}$, $X_{\beta,\sigma,n}$ exponentially converges to $0$ $\mathbb{P}$-a.s.\ as $n \to \infty$.
\end{prop}

\begin{proof}
Fix $(\beta,\sigma) \in \tilde{\rD}_{\lambda<0}$. By a stable manifold theorem~\cite[Proposition~3.3]{Malicet2017}, for $\mu_{\beta,\sigma}$-almost all $(x,\theta) \in \bbR \times [0,2\pi)$, there exists $\delta>0$ such that
\[ f_{\beta,\sigma,\theta,n}(x+\delta) - f_{\beta,\sigma,\theta,n}(x-\delta) \to 0 \textrm{ exponentially as } n \to \infty. \]
So take a value $\delta_1>0$ sufficiently small that the set
\[ Y := \left\{ (x,\theta) \, : \, f_{\beta,\sigma,\theta,n}(x+\delta_1) - f_{\beta,\sigma,\theta,n}(x-\delta_1) \to 0 \textrm{ exponentially as } n \to \infty \right\} \]
has $\mu_{\beta,\sigma}(Y)>0$. Let
\[ Z = \left\{ (x,\theta) \in \bbR \times [0,2\pi) \, : \, \lim_{n \to \infty} \frac{1}{N} \#\left\{n \in \{1,\ldots,N\} : G_{\beta,\sigma}^n(x,\theta) \in Y \right\} =
\mu_{\beta,\sigma}(Y) \right\}. \]
Since $\mu_{\beta,\sigma}$ is a $G_{\beta,\sigma}$-ergodic probability measure, we have that $\mu_{\beta,\sigma}(Z)=1$. Now since $X_{\beta,\sigma,n}$ converges in probability to $0$ as $n \to \infty$, by Lemma~\ref{lemma:App_convprob} (applied to $\mathbbm{1}_{\{X_{\beta,\sigma,n} > \delta_1\}}$) there is a $\mathbb{P}$-full measure set $E_0 \subset [0,2\pi)$ such that for all $\theta \in E_0$,
\[ \limsup_{n \to \infty} \frac{1}{N} \#\left\{n \in \{1,\ldots,N\} : X_{\beta,\sigma,n}(\theta) \leq \delta_1 \right\} =
1. \]
Let $E$ be the intersection of $E_0$ and the $\theta$-projection of $Z$. So $E$ is a $\mathbb{P}$-full measure set. We will show that for all $\theta \in E$, $X_{\beta,\sigma,n}(\theta) \to 0$ exponentially as $n \to \infty$. Fix $\theta \in E$, and fix $x$ such that $(x,\theta) \in Z$. Since
\begin{align*}
    \limsup_{n \to \infty} \frac{1}{N} \#\left\{n \in \{1,\ldots,N\} : X_{\beta,\sigma,n}(\theta) \leq \delta_1 \right\} &= 1 \quad \textrm{and} \\
    \lim_{n \to \infty} \frac{1}{N} \#\left\{n \in \{1,\ldots,N\} : G_{\beta,\sigma}^n(x,\theta) \in Y \right\} &> 0,
\end{align*}
there must exist $n_0 \geq 1$ with the properties that $X_{\beta,\sigma,n_0}(\theta) \leq \delta_1$ and $G_{\beta,\sigma}^{n_0}(x,\theta) \in Y$. But
\begin{itemize}
    \item the fact that $X_{\beta,\sigma,n_0}(\theta) \leq \delta_1$ implies that
    \[ f_{\beta,\sigma,\theta,n_0}(\overline{\bbR}) \subset [f_{\beta,\sigma,\theta,n_0}(x)-\delta_1,f_{\beta,\sigma,\theta,n_0}(x)+\delta_1] =: I_{n_0} \]
    and therefore
    \begin{equation} \label{eq:App_In0} f_{\beta,\sigma,\theta,n_0+n}(\overline{\bbR}) \subset f_{\beta,\sigma,g^{n_0}(\theta),n}(I_{n_0}) \end{equation}
    for all $n \geq 0$;
    \item the fact that $G_{\beta,\sigma}^{n_0}(x,\theta) \in Y$ implies that
    \begin{equation} \label{eq:App_In0contr} \mathrm{diam}(f_{\beta,\sigma,g^{n_0}(\theta),n}(I_{n_0})) \to 0 \textrm{ exponentially as } n \to \infty. \end{equation}
\end{itemize}
Combining Eqs.~\ref{eq:App_In0} and \eqref{eq:App_In0contr} gives that $X_{\beta,\sigma,n}(\theta) = \mathrm{diam}(f_{\beta,\sigma,\theta,n}(\overline{\bbR})) \to 0$ exponentially as $n \to \infty$.
\end{proof}

\begin{prop} \label{prop:App_pi/2}
    For all $\sigma \geq 0$, there is no $\mathbb{P}$-full measure subset of $[0,2\pi)$ on which $X_{0,\sigma,n} \to 0$ uniformly as $n \to \infty$.
\end{prop}

\begin{proof}
The set $\{\frac{\pi}{2},\frac{3\pi}{2}\}$ is a $2$-periodic orbit of $g$ on which $\cos=0$, and hence
\[ f_{0,\sigma,\frac{\pi}{2},n} = f_0^n \]
for all $n \geq 0$. Hence
\[ \lim_{n \to \infty} X_{0,\sigma,n}(\tfrac{\pi}{2}) = \left(\lim_{n \to \infty} f_0^n(\infty) \right) - \left(\lim_{n \to \infty} f_0^n(-\infty) \right) = 2a_0 > 0. \]
Since $X_{0,\sigma,n}$ is continuous for each $n$, and $\tfrac{\pi}{2}$ belongs to the closure of every $\mathbb{P}$-full measure set, we then have the desired result.
\end{proof}

\subsection{Continuous dependence of $a$ and $\mu$} \label{sec:App_ctsD}

Let $\mathfrak{S}=\{(\beta,\sigma,\boldsymbol{\theta}) \in \bbR \times [0,\infty) \times \Theta \, : \, 
\boldsymbol{\theta} \in \Theta_{\beta,\sigma} \}$. Then exactly the same argument as in Sec.~\ref{sec:App_cont} with $\mathfrak{S}$ in place of $S^+ \times \Theta$ and with $\overline{\bbR}$ in place of $[x^*,\infty]$ yields that the map $(\beta,\sigma,\boldsymbol{\theta}) \mapsto a_{\beta,\sigma}(\boldsymbol{\theta})$ is continuous on $\mathfrak{S}$.

Consequently, for each $(\beta,\sigma) \in \tilde{\rD}$, the image of $\Theta_{\beta,\sigma}$ under $\Pi_{\beta,\sigma}$ is contained in the support $\cA_{\beta,\sigma}$ of $\mu_{\beta,\sigma}$.

Moreover (similarly to in Sec.~\ref{sec:App_ergcts}), over $(\beta,\sigma) \in \tilde{\rD}$, $\mu_{\beta,\sigma}$ depends continuously on $(\beta,\sigma)$ in the topology of weak convergence: Fix a sequence $(\beta_n,\sigma_n)$ in $\tilde{\rD}$ converging to a point $(\beta,\sigma) \in \tilde{\rD}$. The set $\Theta_{\beta,\sigma} \cap \bigcap_n \Theta_{\beta_n,\sigma_n}$ is a $\tilde{\mathbb{P}}$-full measure set on which $\Pi_{\beta_n,\sigma_n}$ converges pointwise to $\Pi_{\beta,\sigma}$, and so the dominated convergence theorem gives that $\mu_{\beta_n,\sigma_n}$ converges weakly to $\mu_{\beta,\sigma}$.

\subsection{Basin of $\cA_{\beta,\sigma}$ and $\mu_{\beta,\sigma}$ in 
$\tilde{\rD}_{\lambda<0}$} \label{sec:App_BasinD}

For $(\beta,\sigma) \in \tilde{\rD}$, we said in the proof of Proposition~\ref{prop:App_tildeD} that $\dot{a}^+_{\beta,\sigma}$ and $\dot{a}^-_{\beta,\sigma}$ fulfil the invariant-graph property; hence, since
\[ \Theta_{\beta,\sigma}=\{\boldsymbol{\theta} \in \Theta \, : \, \dot{a}^+_{\beta,\sigma}(\boldsymbol{\theta}) = \dot{a}^-_{\beta,\sigma}(\boldsymbol{\theta}) \}, \]
we have that $g(\Theta_{\beta,\sigma}) \subset \Theta_{\beta,\sigma}$ and $a_{\beta,\sigma}$ fulfils the invariant-graph property: for all $(\theta_i)_i \in \Theta_{\beta,\sigma}$ and $n \geq 0$,
\begin{equation} \label{eq:App_a-inv} a_{\beta,\sigma}((\theta_{i+n})_i) = f_{\beta,\sigma,\theta_0,n}\big( a_{\beta,\sigma}((\theta_i)_i) \big). \end{equation}
Hence, a consequence of Proposition~\ref{prop:App_expconv} is the following.

\begin{cor} \label{cor:App_attrinv}
    Fix $(\beta,\sigma) \in \tilde{\rD}_{\lambda<0}$. There is a $\tilde{\mathbb{P}}$-full measure set $\tilde{E} \subset \Theta_{\beta,\sigma}$ such that for every $(x,(\theta_i)_i) \in \bbR \times \tilde{E}$,
    \[ \left| f_{\beta,\sigma,\theta_0,n}(x) - a_{\beta,\sigma}((\theta_{i+n})_i) \right| \to 0 \quad \textrm{as } n \to \infty. \]
\end{cor}

\begin{proof}
On the basis of Proposition~\ref{prop:App_expconv}, let $E \subset [0,2\pi)$ be a $\mathbb{P}$-full measure set on which $X_{\beta,\sigma,n}$ converges pointwise to $0$. Let $\tilde{E}$ be the set of points in $\Theta_{\beta,\sigma}$ whose $0$-coordinate belongs to $E$. Then due to \eqref{eq:App_a-inv}, for all $(x,(\theta_i)_i) \in \bbR \times \tilde{E}$ we have
\[ \left| f_{\beta,\sigma,\theta_0,n}(x) - a_{\beta,\sigma}((\theta_{i+n})_i) \right| \leq \mathrm{diam}(f_{\beta,\sigma,\theta_0,n}(\bbR)) = X_{\beta,\sigma,n}(\theta_0) \to 0 \quad \textrm{as } n \to \infty. \qedhere \]
\end{proof}

\begin{cor} \label{cor:App_Dbasin}
    Fix $(\beta,\sigma) \in \tilde{\rD}_{\lambda<0}$. There is a $\mathbb{P}$-full measure set $E_0 \subset [0,2\pi)$ such that $\bbR \times E_0$ is contained in the basin of $\cA_{\beta,\sigma}$.
\end{cor}

\begin{proof}
Let $\tilde{E}$ be as in Corollary~\ref{cor:App_attrinv}, and let $E_0$ be the $\theta_0$-projection of $\tilde{E}$. Fix $(x,\theta) \in \bbR \times E_0$, and take a point $(\theta_i)_i \in \tilde{E}$ with $0$-coordinate $\theta$. Since, the image of $\Theta_{\beta,\sigma}$ under $\Pi_{\beta,\sigma}$ is contained in $\cA_{\beta,\sigma}$, we have in particular that
\[ \Big( a_{\beta,\sigma}((\theta_{i+n})_i) \, , \, g^n(\theta) \Big) \in \cA_{\beta,\sigma}. \]
for all $n \geq 0$. So then, from Corollary~\ref{cor:App_attrinv}, we have that the distance of $G^n(x,\theta)$ from $\cA_{\beta,\sigma}$ tends to $0$ as $n \to \infty$.
\end{proof}

\begin{prop} \label{prop:App_physD}
    Fix $(\beta,\sigma) \in \tilde{\rD}_{\lambda<0}$. There is a $\mathbb{P}$-full measure set $R \subset [0,2\pi)$ such that $\bbR \times R$ is contained in the basin of $\mu_{\beta,\sigma}$.
\end{prop}

\begin{proof}
Construct $\Theta_1 \subset \Theta$ as in the proof of Proposition~\ref{prop:App_phys} with $a_{\beta,\sigma}$ in place of $a^+_{\beta,\sigma}$. Let $\Theta_2=\Theta_1 \cap \tilde{E}$, where $\tilde{E}$ is as in Corollary~\ref{cor:App_attrinv}. Let $R$ be the $\theta_0$-projection of $\Theta_2$. Then for any $(x,\theta) \in \bbR \times R$, one shows that $(x,\theta)$ is in the basin of $\mu_{\beta,\sigma}$ by the same argument as in the proof of Proposition~\ref{prop:App_phys}, using $\Theta_2$ in place of $\Theta_1$ and using Corollary~\ref{cor:App_attrinv} in place of Lemma~\ref{lemma:App_contr}.
\end{proof}

\end{document}